\documentclass[11pt,tightenlines,nofootinbib,superscriptaddress]{revtex4-2}

\usepackage[toc,page,header]{appendix}
\usepackage{minitoc}
\usepackage[utf8]{inputenc}
\usepackage[british]{babel}
\usepackage{amsmath,amssymb,amsthm,amscd}
\usepackage{mathtools}
\usepackage{mathrsfs}
\usepackage[active]{srcltx}
\usepackage{pgf,tikz}
\usepackage{mathrsfs}
\usepackage{enumerate}
\usepackage{braket}
\usetikzlibrary{arrows}
\usepackage{dsfont}
\usepackage{paralist}
\usepackage{comment}
\usepackage{enumerate}
\usepackage{natbib}
\usepackage{soul}

\definecolor{navyblue}{rgb}{0.0, 0.0, 0.5}
\usepackage[colorlinks=true, urlcolor=blue,citecolor=purple,anchorcolor=blue, linkcolor=black]{hyperref}
\usepackage[margin=2.54cm,bmargin=2.54cm,tmargin=2.54cm]{geometry}
\usepackage{todonotes}
\usepackage{cleveref}

\newcommand{\tr}{\operatorname{tr}}

\newcommand{\cH}{\mathcal{H}}

\newcommand{\norm}[1]{\left\lVert#1\right\rVert}

\renewcommand{\phi}{\varphi}
\renewcommand{\epsilon}{\varepsilon}

\newcommand{\JDW}[1]{\textcolor{red}{(JDW: #1)}}
\newcommand{\DSF}[1]{\textcolor{orange}{(DSF: #1)}}
\theoremstyle{plain}
\newtheorem{theorem}{Theorem}[section]
\newtheorem{corollary}[theorem]{Corollary}
\newtheorem{lemma}[theorem]{Lemma}
\newtheorem{proposition}[theorem]{Proposition}

\newtheorem{definition}[theorem]{Definition}

\newtheorem{condition}[theorem]{Condition}

\theoremstyle{remark}
\newtheorem{remark}[theorem]{Remark}

\newtheorem*{claim*}{Claim}
\newtheorem*{remark*}{Remark}
\newtheorem*{example*}{Example}
\newtheorem*{notation*}{Notation}
\numberwithin{equation}{section}

\makeatletter
\@namedef{subjclassname@2020}{\textup{2020} Mathematics Subject Classification}
\makeatother

\begin{document}

	\author{\begingroup
		\hypersetup{urlcolor=navyblue}
		\href{https://orcid.org/0000-0003-1610-1597}{Emilio Onorati}
		\endgroup}
	\email[Emilio Onorati ]{emilio.onorati@tum.de}

	\author{\begingroup
		\hypersetup{urlcolor=navyblue}
		\href{https://orcid.org/0000-0001-7712-6582}{Cambyse Rouz\'{e}
			\endgroup}
	}
	\email[Cambyse Rouz\'{e} ]{cambyse.rouze@tum.de}
	\affiliation{Zentrum Mathematik, Technische Universit\"{a}t M\"{u}nchen, 85748 Garching, Germany}

	\author{\begingroup
		\hypersetup{urlcolor=navyblue}
		\href{https://orcid.org/0000-0001-9699-5994}{Daniel Stilck Fran\c{c}a}
		\endgroup}
	\affiliation{Univ Lyon, ENS Lyon, UCBL, CNRS, Inria, LIP, F-69342, Lyon Cedex 07, France}
	\email[Daniel Stilck Fran\c ca ]{daniel.stilck\_franca@ens-lyon.fr}

	\author{\begingroup
		\hypersetup{urlcolor=navyblue}
		\href{https://orcid.org/0000-0002-6077-4898}{James D. Watson}
		\endgroup}
	\affiliation{University of Maryland, College Park, QuICS 3353 Atlantic Building, MD 20742-2420, USA
	}
	\email[James D. Watson ]{jdwatson@umd.edu}

	%\author{Anonymous Authors}
	
	%\title[]{Efficient learning of ground and thermal states of  lattice systems}
	
	\title[]{Efficient learning of ground \& thermal states within phases of matter }
	
	% \title[]{Efficient learning of lattice quantum systems and phases of matter}

	\begin{abstract}
		
		We consider two related tasks: (a) estimating a parameterisation of an unknown Gibbs state and expectation values of Lipschitz observables on this state; and (b) learning the expectation values of local observables within a thermal or quantum phase of matter. In both cases, we wish to minimise the number of samples we use to learn these properties to a given precision.
		
		For the first task, we develop new techniques to learn parameterisations of classes of systems, including quantum Gibbs states of non-commuting Hamiltonians under the condition of exponential decay of correlations and the approximate Markov property, thus improving on work by \cite{rouze2021learning}. We show that it is possible to infer the expectation values of all extensive properties of the state from a number of copies that not only scales polylogarithmically with the system size, but polynomially in the observable's locality --- an exponential improvement over state-of-the-art  --- hence partially answering conjectures stated in \cite{rouze2021learning} and \cite{anshu2021sample} in the positive. This class of properties includes expected values of quasi-local observables and entropic quantities of the state. 
		
		For the second task, we turn our tomography tools into efficient algorithms for learning observables in a phase of matter of a quantum system. 
		By exploiting the locality of the Hamiltonian, we show that $M$ local observables can be learned with probability $1-\delta$ and up to precision $\epsilon$ with access to only $N=\mathcal{O}\big(\log\big(\frac{M}{\delta}\big)e^{\operatorname{polylog}(\epsilon^{-1})}\big)$ samples ---
		again an exponential improvement in the precision over the best previously known bounds \cite{huang2021provably}. 
		Our results apply to both thermal phases of matter displaying exponential decay of correlations and families of ground states of Hamiltonians satisfying a similar condition.
		In addition, our sample complexity applies to the worse case setting whereas previous results only applied to the average case setting.
		
		To prove our results, we develop new tools of independent interest, such as robust shadow tomography algorithms for ground and Gibbs states, Gibbs approximations of locally indistinguishable ground states, and generalisations of transportation cost inequalities for Gibbs states of non-commuting Hamiltonians.

	\end{abstract}
	
	\maketitle
	\doparttoc % Tell to minitoc to generate a toc for the parts
	
	\faketableofcontents % Run a fake tableofcontents command for the partocs

	\newpage
	
	\section{Introduction}
	Tomography of quantum states is among the most important tasks in quantum information science. In quantum tomography, we have access to one or more copies of a quantum state and wish to understand the structure of the state. 
	However, for a general quantum state, all tomographic methods inevitably require resources that scale exponentially in the size of the system~\cite{Haah_2017,optimal_wright}. 
	This is due to the curse of dimensionality: the number of parameters needed to fully describe a
	quantum system scales exponentially with the number of its constituents. 
	Obtaining these parameters often necessitates the preparation and destructive measurement of exponentially many copies of the quantum system, as well as their storage in a classical memory.
	In particular, as the size of quantum devices continues to increase beyond what can be easily simulated classically, the community faces new challenges to characterise their output states in a robust and efficient manner.
	
	Thankfully, only a few physically relevant observables are often needed to describe the physics of a system, e.g.~its entanglement or energy. 
	Recently, new methods of tomography have been proposed which precisely leverage this important simplification to develop efficient state learning algorithms. One highly relevant development in this direction is that of classical shadows \cite{Huang2020}. 
	This new set of protocols allows for estimating physical observables of quantum spin systems that only depend on local properties from a number of measurements that scales logarithmically with the total
	number of qubits.
	However, the number of required measurements still faces an exponential growth with respect to the size of the observables that we want to estimate. 
	Thus, using such protocols to learn the expectation values of physical observables that depend on more than a few qubits quickly becomes unfeasible.

	%As the size of quantum devices continues to increase beyond
	%what can be easily simulated classically, the community faces new challenges to characterize their output states in a robust
	%and efficient manner. This is due to the
	%curse of dimensionality: the number of parameters needed to fully describe a
	%quantum system scales exponentially with the number of its constituents. Obtaining these parameters often necessitates
	%the preparation and destructive measurement of exponentially many copies of the quantum system, as well as their storage
	%in a classical memory. Thankfully, only
	%few physically relevant observables are often needed to describe the physics of a system, e.g.~its entanglement or energy. Recently, new methods of tomography have been proposed which precisely
	%leverage this important simplification to develop efficient state learning algorithms. One highly relevant development in this
	%direction is that of
	%classical shadows \cite{Huang2020,huang2021provably}. This new set of protocols allows for estimating physical observables of quantum
	%spin systems that only depend on local properties from a number of measurements that scales logarithmically with the total
	%number of qubits. However, the number of required measurements still faces an exponential growth with respect to the size of the observables that we want to estimate. 
	%Thus, using such protocols to learn the expectation values of physical observables that depend
	%on more than a few qubits quickly becomes unfeasible. 

	\medskip
	\textbf{Gibbs State Tomography.} \ 
	Some simplification can be achieved from the fact that physically relevant quantum states, such as ground and Gibbs states of a locally interacting spin system, are themselves often described by a number of parameters which scales only polynomially
	with the number of qubits. 
	From this observation, another direction in the characterisation of large quantum systems that has received considerable attention is that of Hamiltonian learning and many-body tomography, where it was recently shown that it is possible to robustly
	characterise the interactions of a Gibbs state with a few samples \cite{anshuweb,haah2021optimal}. 
	However, even for many-body states, recovery in terms of the trace distance requires a number of samples that scales polynomially in the number of qubits, in contrast to shadows for which the scaling is logarithmic.

	These considerations naturally lead to the question of identifying settings where it is possible to combine the strengths of shadows and many-body tomography. In \cite{rouze2021learning}, the authors proposed a first solution by combining these with new insights from the emerging field of
	quantum optimal transport. They obtained a tomography algorithm that only requires a number of samples that scales logarithmically in the system's size and learns all quasi-local properties of a state.
	These properties are characterised by so-called ``Lipschitz observables''.
	However, that first step was confined to topologically trivial states such as high-temperature Gibbs states of commuting Hamiltonians or outputs of shallow circuits. Here, we significantly extend these results to all states exhibiting exponential decay of correlations and the approximate Markov property.
	% and for which efficient Hamiltonian learning algorithms exist.
	
	\medskip
	\textbf{Learning Phases of Matter.} \ Tomographical techniques by themselves are somewhat limited in that they tell us nothing about nearby related states – often states belong to a phase of matter in which the properties of the states vary smoothly and are in some sense ``well behaved’’, and we wish to learn properties of this entire phase of matter. 
	A recent line of research in this direction that has gained significant attention from the quantum community is that of combining machine learning methods with the ability to sample complex quantum states from a phase of matter to efficiently characterise the entire phase~\cite{Biamonte2017,Carrasquilla2017}, as well as using ML techniques to improve identifying phases of matter \cite{rodriguez2019identifying,rem2019identifying,Dong_Pollmann_Zhang_2019} and approximating quantum states \cite{gao2017efficient,park2020geometry,barr2020quantum, nomura2021purifying}.
	It is well known that these tasks are computationally intractable in general \cite{ambainis2014physical,watson2021complexity, bravyi2022quantum}, and so having access to data from an externally generated source could conceivably speed up these computations.
	A landmark result in this direction is~\cite{huang2021provably}. 
	There the authors showed how to use machine learning methods combined with classical shadows to learn local linear and nonlinear functions of states belonging to a gapped phase of matter with a number of samples that only grows logarithmically with the system's size. 
	That is, given states from that phase drawn from a distribution and the corresponding parameters of the Hamiltonian, one can train a classical algorithm that would predict local properties of other points of the phase. 
	However, there are some caveats to this scheme: (i)~the scaling of the number of samples in terms of the precision is exponential, (ii)~it does not immediately apply to phases of matter beyond gapped ground states, (iii)~the results only come with guarantees on the errors in the prediction in expectation. That is, given another state sampled from the same distribution as the one used to train, only \emph{on average} is the error made by the ML algorithm proven to be small.
	
	In this work, we address all of these shortcomings. 
	First, our result extends to thermal phases of matter which exhibit exponential decay of correlations, which includes all thermal systems away from criticality/poles in the partition function \cite[Section 5]{harrow2020classical}. 
	Our result also extends to phases that satisfy a generalised version local topological quantum order~\cite{michalakis2013stability,bravyi2010topological,nachtergaele2022quasi}. 
	Furthermore, the sample complexity of our algorithm is quasi-polynomial in the desired precision, which is an exponential improvement over previous work \cite{huang2021provably}. 
	And, importantly, it comes with point-wise guarantees on the quality of the recovery, as opposed to average guarantees.
	
	Interestingly, our results are easier to grasp through the lens of the concentration of measure phenomenon rather than machine learning: we show that local expectation values of quantum states are quite smooth under perturbations in the same class of states. And, as is showcased by the concentration of measure phenomenon, smooth functions on high-dimensional spaces do not show a lot of variability. Thus, it suffices to collect a few examples to be able to predict what happens in the whole space, while the price we pay for these stronger recovery guarantees
	% \footnote{During the completion of this work we became aware of work by Huang et al.~\cite{Lewis_Huang_Preskill_2023} that also achieves a quasi-polynomial scaling in precision and can be extended to thermal phases. 
		% However, that work also has average recovery guarantees, see \Cref{sec:previous_work} for a detailed discussion.} 
	is that our algorithm does not work for any distribution over states, but needs some form of anti-concentration which holds e.g.~for the uniform distribution (see \Cref{seclearningalgogibbs} for a technical discussion). In other words, our algorithm necessitates to ``see'' enough of the space to work and struggles if there are large, low-probability corners.
	% On the other hand, recent concurrent work~\cite{Lewis_Huang_Preskill_2023} leverages machine learning techniques to get rid of any assumption on the distribution, but at this time at the price of weaker recovery guarantees.
	
	\section{Summary of main results}

	% \subsection{Preliminaries}

	In this paper, we consider a quantum system defined over a $D$-dimensional finite regular lattice $\Lambda=[-L,L]^D$, where $n=(2L+1)^D$ denotes the total number of qubits constituting the system. We assume for simplicity that each site of the lattice hosts a qubit, so that the total system's Hilbert space is $\cH_\Lambda:=\bigotimes_{j\in \Lambda} \mathbb{C}^{2}$, although all of the results presented here easily extend to qudits.
	
	Our focus in this work are nontrivial statements about what can be learned about many-body states of $n$ qubits in the setting where we are only given $\Theta( \operatorname{polylog}(n))$ copies. 
	The common theme is that we will assume exponential decay of correlations for our class of states, but will show results in two different regimes. In~\Cref{sec:learning_states} we summarise our results on how to estimate \emph{all quasi-local properties of a given state given identical copies of it.} 
	This is the traditional setting of quantum tomography. 
	In contrast, in~\Cref{sec:learning_systems} we summarise our findings on how to learn \emph{local properties of a class of states given samples from different states from that class}. This is the setting of~\cite{huang2021provably} where ground states of gapped quantum phases of matter were studied.
	Here we consider (a) thermal phases of matter with exponentially decaying correlations and (b) ground states satisfying what we call `Generalised Approximate Local Indistinguishability' (GALI).
	Using techniques used in \cite{Lewis_Huang_Preskill_2023}, we show this includes gapped ground states.
	
	\subsection{Optimal Tomography of Many-Body Quantum States}~\label{sec:learning_states}
	
	We first consider the task of obtaining a good approximation of expected values of extensive properties of a fixed unknown $n$-qubit state over $\Lambda$. The state is assumed to be a Gibbs state of an unknown local Hamiltonian $H(x):=\sum_{j\in\Lambda}h_j(x_j)$, $x=\{x_j\}\in [-1,1]^m$, defined through interactions $h_j(x_j)$, each depending on parameters $x_j\in[-1,1]^\ell$ for some fixed integer $\ell$ and supported on a ball $A_j$ around site $j\in\Lambda $ of radius $r_0$. We also assume that the matrix-valued functions $x_j\mapsto h_j(x_j)$ as well as their derivatives are uniformly bounded: $\|h_j\|_\infty,\|\nabla h_j\|_\infty\le h$. The corresponding Gibbs state at inverse temperature $\beta>0$, and the ground state as $\beta\to\infty$ take the form 
	\begin{align}\label{Gibbsgrounddef}
		\sigma(\beta,x):=\frac{e^{-\beta H(x)}}{\tr\big[e^{-\beta H(x)}\big]}\qquad \quad  \text{ and }\qquad \quad  \psi_g(x):=\lim_{\beta\to\infty}\sigma(\beta,x)\,.
	\end{align}
	In the case when $[h_j(x_j),h_{j'}(x_{j'})]=0$ for all $j,j'\in\Lambda$, the Hamiltonian $H(x)$ and its associated Gibbs states $\sigma(\beta,x)$ are said to be commuting.

	\subsubsection{Preliminaries on Lipschitz observables}\label{Lipsec}
	Extensive properties of a state are well-captured by the recently introduced class of Lipschitz observables~\cite{Rouz2019,de2021quantum}. 
	
	\begin{definition}[Lipschitz Observable \cite{de2021quantum} ]
		An observable $L$ on $\mathcal{H}_{\Lambda}$ is said to be \emph{Lipschitz} if $\|L\|_{\operatorname{Lip}}:= \max_{i\in\Lambda}\min_{L_{i^c}}\,2\|L-L_{i^c}\otimes I_i\|_\infty=\mathcal{O}(1)$, where $i^c$ is the complement of the site $i$ in $\Lambda$ and the scaling is in terms of the number of qubits in the system.
	\end{definition}
	
	In words, $\|L\|_{\operatorname{Lip}}$ quantifies the amount by which the expectation value of $L$ changes for states that are equal when tracing out one site. By a simple triangle inequality together with \cite[Proposition 15]{de2021quantum}, one can easily see that $\|L\|_\infty\le n\|L\|_{\operatorname{Lip}}$. Given the definition of the Lipschitz constant, we can also define the quantum  Wasserstein distance of order $1$ by duality~\cite{de2021quantum}.
	\begin{definition}[Wasserstein Distance~\cite{de2021quantum}] \label{Def:Wasserstein_Distance}
		The \emph{Wasserstein distance} between two $n$ qubit quantum states $\rho_1,\rho_2$ is defined as $W_1(\rho_0,\rho_1):=\sup_{\|L\|_{\operatorname{Lip}}\le 1}\,\tr \big[L(\rho_0-\rho_1)\big]$. It satisfies $W_1(\rho_0,\rho_1)\le n\|\rho-\sigma\|_1$.
	\end{definition}

	Having $W_1(\rho,\sigma)=\mathcal{O}(\epsilon n)$ is sufficient to guarantee that the expectation value of $\rho$ and $\sigma$ on \emph{extensive, quasi-local} observables is the same up to a \emph{multiplicative} error $\epsilon n$. This fact justifies why we focus on learning states up to an error $\mathcal{O}(\epsilon n)$ in Wasserstein distance instead of the usual trace distance bound of order $\mathcal{O}(\epsilon)$: although a trace distance guarantee of order $\mathcal{O}(\epsilon)$ would give the same error estimate, it requires exponentially more samples even for product states, as shown in~\cite[Appendix G]{rouze2021learning}. In \Cref{AppendixLip}, we argue that Lipschitz observables and the induced Wasserstein distance capture both linear and nonlinear extensive properties of many-body quantum states.

	\subsubsection{Gibbs state tomography}\label{Gibbsstatetomo}

	In this section, we turn our attention to the problem of obtaining approximations of linear functionals of the form
	$f_L(\beta ,x):=\tr [L\sigma(\beta ,x)]$ for all Lipschitz observables $L$ from the measurement and classical post-processing of as few copies of the associated unknown Gibbs state $\sigma(\beta,x)$ as possible. We will further require that the state satisfies the property of exponential decay of correlations: for any two observables $X_A$, resp.~$X_B$, supported on region $A$, resp.~$B$,
	\begin{align}\label{decaycorrintro}
		{\operatorname{Cov}_{\sigma(\beta,x)}(X_A,X_B)}\le C\,\min\{|A|,|B|\}\,{\|X_A\|_\infty\,\|X_B\|_\infty}\,e^{-\nu \operatorname{dist}(A,B)}\,,
	\end{align}
	for some constants $C,\nu>0$, where $\operatorname{dist}(A,B)$ denotes the distance between regions $A$ and $B$, and where the covariance is defined by 
	\begin{align}
		\operatorname{Cov}_{\sigma}(X,Y):=\frac{1}{2}\,\tr\big[\sigma \,\big\{X-\tr[\sigma X],Y-\tr[\sigma Y]\big\}\big]\,.
	\end{align}
	
	Our first main result is a method to learn Gibbs states with few copies of the unknown state:
	\begin{theorem}[Tomography algorithm for decaying Gibbs states (informal)]\label{maintheoremtomo}
		For any unknown commuting Gibbs state $\sigma(\beta,x)$ satisfying \Cref{decaycorrintro}, there exists an algorithm that provides the description of parameters $x'$ such that the state $\sigma(\beta,x')$ approximates $\sigma(\beta,x)$ to precision $n\epsilon$ in Wasserstein distance with probability $1-\delta$ with access to $N=\mathcal{O}\big(\log(\delta^{-1})\operatorname{polylog}(n)\,\epsilon^{-2}\big)$ samples of the state (see \Cref{seccommuting}). The result extends to non-commuting Hamiltonians whenever one of the following two assumptions is satisfied: 
		\begin{enumerate}[(i)]
			\item  the high-temperature regime, $\beta<\beta_c$ (see \Cref{sechightemperature}).
			\item uniform clustering/Markov conditions (see \Cref{coroapproxmarkov}).
		\end{enumerate}
		In case (ii), we find good approximation guarantees under the following slightly worst scaling in the precision $\epsilon$: $N=\mathcal{O}(\epsilon^{-4}\operatorname{polylog}(n\delta^{-1}))$.
	\end{theorem}

	\noindent \textbf{Proof ideas}: The results for commuting Hamiltonians and in the high-temperature regime proceed directly from the following continuity bound on the Wasserstein distance between two arbitrary Gibbs states, whose proof requires the notion of quantum belief propagation in the non-commuting case (see \Cref{cor:continuity_W1}): for any $x,y\in[-1,1]^m$,
	\begin{align}\label{upperboundW1}
		W_1(\sigma(\beta,x),\sigma(\beta,y))=  \,\|x-y\|_{\ell_1}\,\mathcal{O}(\operatorname{polylog}(n))\,.
	\end{align}
	Furthermore, this inequality is tight up to a $\operatorname{polylog}(n)$ factor for $\beta=\Theta(1)$. \Cref{upperboundW1} reduces the problem of recovery in Wasserstein distance to that of recovering the parameters $x$ up to an error $\epsilon n/\operatorname{polylog}(n)$ in $\ell_1$ distance. This is a variation of the Hamiltonian learning problem for Gibbs states~\cite{anshu2021sample,haah2021optimal} which relies on lower bounding the $\ell_2$ strong convexity constant for the log-partition function.
	
	In~\cite{anshuweb}, the authors give an algorithm estimating~$x$ with  $e^{\mathcal{O}(\beta k^D)}\mathcal{O}(\log(\delta^{-1}n)\epsilon^{-2})$ copies of $\sigma(\beta,x)$ up to $\epsilon$ in $\ell_\infty$ distance when $\sigma(\beta,x)$ belongs to a family of commuting, $k$-local Hamiltonians on a $D$-dimensional lattice. If we assume $m=\mathcal{O}(n)$, this translates to an algorithm with sample complexity $e^{\mathcal{O}(\beta k^D)}\mathcal{O}(\epsilon^{-2}\textrm{polylog}(\delta^{-1}n))$ to learn $x$ up to $\epsilon n$ in $\ell_1$ distance. It should also be noted that the time complexity of the algorithm in \cite{anshuweb} is $\mathcal{O}(ne^{\mathcal{O}(\beta k^D)}\epsilon^{-2}\textrm{polylog}(\delta^{-1}n))$.
	Thus, any commuting model at constant temperature satisfying exponential decay of correlations can be efficiently learned with $\operatorname{polylog}(n)$ samples. We refer the reader to \Cref{sec:beyond_high_T} for more information and classes of commuting states that satisfy exponential decay of correlations.
	In the high-temperature regime, we rely on a result of~\cite{haah2021optimal} where the authors give a computationally efficient algorithm to learn $x$ up to error $\epsilon$ in $\ell_\infty$ norm from $\mathcal{O}(\epsilon^{-2}\textrm{polylog}(\delta^{-1}n))$ samples. This again translates to a $\mathcal{O}(\epsilon n)$ error in $\ell_1$ norm thanks to \eqref{upperboundW1}. 
	
	Furthermore, in \Cref{sec:gibbs_learning_exponential_decay} we more directly extend the strategy of \cite{anshu2021sample} by introducing the notion of a $W_1$ strong convexity constant for the log-partition function and showing that it scales linearly with the system size under (a) uniform clustering of correlations and (b) uniform Markov condition. This result also generalises the strategy of \cite{rouze2021learning} which relied on the existence of a so-called transportation cost inequality previously shown to be satisfied for commuting models at high-temperature. For the larger class of states satisfying conditions (a) and (b), we are able to find $x'$ s.t.  $W_1(\sigma(\beta,x),\sigma(\beta,x'))=\mathcal{O}(\epsilon n)$ with $\mathcal{O}(\epsilon^{-4}\textrm{polylog}(\delta^{-1}n))$ samples. Note that the uniform Markov condition is expected to hold for a large class of models that goes beyond high-temperature Gibbs states \cite{kato2019quantum,kuwahara_gibbs}.
	%Note that their algorithm is also shown to be computationally efficient. 
	We believe that our result is also of interest for classical models. The last years have seen a flurry of results on learning classical Gibbs states under various metrics, with a particular focus on learning the parameters in some $\ell_p$ norm~\cite{Bresler2015,Lokhov2018,klivans2017learning,10.5555/3454287.3455012,one_or_multiple}. But, to the best of our knowledge, the learning in $W_1$ was not considered, particularly the quantum version of the Wasserstein distance. Furthermore, there are phases of classical Ising models that exhibit exponential decay of correlations but no polynomial-time algorithms to sample from the underlying Gibbs states are known~\cite{pirogov_sinai,Lubetzky2013}. This shows that the classes of states our result applies to goes beyond those for which there exist efficient classical algorithms to compute local properties.

	\subsubsection{Beyond linear functionals}\label{beyondlin}

	So far, we considered properties of the quantum system which could be related to local linear functionals of the unknown state. In \cite{Huang2020,huang2021provably}, the authors propose a simple trick in order to learn non-linear functionals of many-body quantum systems, e.g.~their entropy over a small subregion. However, such methods require a number of samples scaling exponentially with the size of the subregion, and thus very quickly become inefficient as the size of the region increases. Here instead, we make use of the continuity of the entropy functional with respect to the Wasserstein distance, mentioned in \Cref{eqentropycontbound}, together with the following Wasserstein continuity bound in order to estimate the entropic quantities of Gibbs states over regions of arbitrary size (see \Cref{localW1bound}): assuming \Cref{decaycorrintro}, for any region $S$ of the lattice and any two $x,y\in[-1,1]^m$ 
	\begin{align}\label{localW1boundestimate}
		W_1(\tr_{S^c}(\sigma(\beta,x)),\tr_{S^c}(\sigma(\beta,y)))\le \|x|_{\mathcal{S}(r_S)}-y|_{\mathcal{S}(r_S)}\|_{\ell_1}\,\operatorname{polylog}(|S(r_S)|)\,,
	\end{align}
	where $r_S=\max\big\{r_0,\, 2\xi\log\big({2|S|C_1}{\|x|_{\mathcal{S}(r_0)}-y|_{\mathcal{S}(r_0)}\|^{-1}_{\ell_1}}\big)\big\}$ with $r_0$ being the smallest integer such that $x|_{\mathcal{S}(r_0)}\ne y|_{\mathcal{S}(r_0)}$, $\mathcal{S}(r_S):=\{x_j|\,\operatorname{supp}(h_j(x_j))\cap S(r_S)\ne \emptyset\}$, $S(r_S):=\{i\in\Lambda:\operatorname{dist}(i,S)\le r_S\}$, and $C_1,\xi>0$ are constants introduced in \Cref{Lemma:Local_Perturbation}.

	% Although we only consider the quantum conditional mutual information
	% von Neumann entropy functional for sake of conciseness, the same idea can be carried over to other entropic quantities such as quantum mutual and conditional mutual information, and R\'{e}nyi variants thereof.
	
	Let us recall a few definitions: denoting by $\rho_R:=\tr_{R^c}(\rho)$ the marginal of a state $\rho\in\mathcal{D}(\cH_\Lambda)$ on a region $R\subset \Lambda$, and given separated regions $A,B,C\subset \Lambda$ of the lattice: $S(A)_\rho:=-\tr[\rho_A\log\rho_A]$ is the von Neumann entropy of $\rho$ on $A$, $S(A|B)_\rho:=S(AB)_\rho-S(B)_\rho$ is the conditional entropy on region $A$ conditioned on region $B$, $I(A:B)_{\rho}:=S(A)_\rho+S(B)_\rho-S(AB)_\rho$ is the mutual information between regions $A$ and $B$, and $I(A:B|C)_\rho:=S(AC)_{\rho}+S(BC)_\rho-S(C)_\rho-S(ABC)_\rho$ is the conditional mutual information between regions $A$ and $B$ conditioned on region $C$. The following corollary is a direct consequence of \Cref{eqentropycontbound} together with \Cref{localW1boundestimate}:

	\begin{corollary}
		Assume the decay of correlations holds uniformly, as specified in \Cref{decaycorrintro}, for all $\{\sigma(\beta,x)\}_{x\in[-1,1]^m}$, $m=\mathcal{O}(n)$. Then, in the notations of the above paragraph, for any two Gibbs states $\sigma(\beta,x)$ and $\sigma(\beta,y)$, $x,y\in[-1,1]^m$, and any region $A\subset \Lambda$:
		\begin{align*}
			&|S(A)_{\sigma(\beta,x)}-S(A)_{\sigma(\beta,y)}|= \,\|x|_{\mathcal{S}(r_S)}-y|_{\mathcal{S}(r_S)}\|_{\ell_1}\mathcal{O}(\operatorname{polylog}(|S(r_S)|))\,,
		\end{align*}
		for $S\equiv A$. The same conclusion holds for $|S(A|B)_{\sigma(\beta,x)}-S(A|B)_{\sigma(\beta,y)}|$ $(S \equiv AB)$, $|I(A:B)_{\sigma(\beta,x)}-I(A:B)_{\sigma(\beta,y)}|$ $(S\equiv AB)$, and $|I(A:B|C)_{\sigma(\beta,x)}-I(A:B|C)_{\sigma(\beta,y)}|$ $(S\equiv ABC)$.
	\end{corollary}

	%\begin{proof}
	%The entropy continuity is a direct consequence of \Cref{localW1boundestimate} together with \Cref{eqentropycontbound}. The rest follows by definition of the entropic quantities considered.
	%\end{proof}
	Thus, given an an estimate $y$ of $x$ satisfying $
	\|x-y\|_{\ell_\infty}=\mathcal{O}(\epsilon/\textrm{polylog}(n))$, we can also approximate entropic quantities of the Gibbs state to a multiplicative error. More generally, entropic continuity bounds can be directly used together with \Cref{maintheoremtomo}(ii) in order to estimate entropic properties of Gibbs states satisfying both uniform clustering of correlations and the approximate Markov condition (see \Cref{sec:gibbs_learning_exponential_decay} for details).

	\subsection{Learning Expectation Values of Parametrised Families of Many-Body Quantum Systems}\label{sec:learning_systems}

	Next, we turn our attention to the task of learning Gibbs or ground states of a parameterised Hamiltonian $H(x)$ known to the learner and sampled according to the uniform distribution $U$ over some region $x\in \Phi \coloneqq \prod_{i=1}^m[-1+x_i^0,1+x_i^0]$ (where here $\prod_{i=1}^m$ represents the Cartesian product over sets).
	More general distributions can also be dealt with under a condition of anti-concentration, see \Cref{seclearningalgogibbs}. 
	Here we restrict our results to local observables of the form
	$O=\sum_{i=1}^M O_i$
	where $S_i:=\operatorname{supp}(O_i)$ is contained in a ball of diameter independent of the system size. The setup in this section is similar to \cite{huang2021provably}. The idea is that we have access to some samples of a state chosen from different values of the parameterised Hamiltonian, and we want to use these to learn observables everywhere in the parameter space with high precision.
	We then want to know: what is the minimum number of samples drawn from this distribution which allows us to accurately predict expectation values of local observables for all choices of parameters?
	
	% From these, we construct an estimator $\hat{f}_O$ of the function $f_O$ as follows: pick $r\in\mathbb{N}$ and denote $S_i(r):=\{j\in\Lambda|\operatorname{dist}(j,S_i)\le r\}$. Then,
	% \begin{align*}
		% &\hat{f}_O(x)=\sum_{i=1}^M\tr \big[O_i\, \sigma(\beta, \hat{Y}_i(x))\big]\,,\quad \text{ with }\\
		% & \hat{Y}_i(x)=\operatorname{argmin}_{{Y}_k}\|x|_{\mathcal{S}_i(r)}-{Y}_k|_{\mathcal{S}_i(r)}\|_{\ell^\infty}\,,
		% \end{align*}
	% where we denote by $x_{\mathcal{S}_{i}(r)}$ the concatenation of vectors $x_j$ corresponding to interactions $h_j$ supported on regions intersecting $S_i(r)$. In words, we approximate the expectation value of $O_i$ by that of the Gibbs state whose parameters in a region around $S_i$ are the closest to the state of interest.

	\subsubsection{Learning Expectation Values in Phases with Exponentially Decaying Correlations}
	
	The learner is given samples $\{(x_i,\sigma(\beta,x_i))\}_{i=1}^N$, where the parameters $x_i\sim U$, and their task is to learn $f_O(x):=\tr[\sigma(\beta,x)O]$ for an arbitrary value of $x\in \Phi$ and an arbitrary local observable $O$. 
	We assume that everywhere in the parameter space $x\in \Phi$ the Gibbs states are in the same phase of exponentially decaying correlations. Note that this does not necessarily imply the existence of a fully polynomial time approximation scheme, and finding under which conditions such algorithms exist is still a very active area of research \cite{pirogov_sinai}.
	
	% \todo{Check this reformulation of the theorem! I'm trying to emphasise the point that we don't need to know the observable $O$ beforehand.}

	\begin{theorem}[Learning algorithm for quantum Gibbs states]\label{learninggibbsalgo}
		With the conditions of the previous paragraph, given a set of $N$ samples $\{x_i,\tilde{\sigma}(\beta, x_i)\}_{i=1}^N$, where $\tilde{\sigma}(\beta,x_i)$ can be stored efficiently classically, and $N = O \big(\log\big(\frac{M}{\delta}\big)\,\log\big(\frac{n}{\delta}\big) e^{\operatorname{polylog}(\epsilon^{-1})} \big)$,
		there exists an algorithm that, on input $x\in\Phi$ and a local observable $O=
		\sum_{i=1}^M O_i$, produces an estimator $\hat{f}_O$
		such that, with probability $(1-\delta)$,
		\begin{align*}
			\sup_{x\in\Phi}\, |f_O(x)-\hat{f}_{O}(x)|\le \epsilon\,\sum_{i=1}^M\|O_i\|_\infty\,.
		\end{align*}
		Moreover, the samples $\tilde{\sigma}(\beta, x_i)$ are efficiently generated from measurements of the Gibbs states $\{\sigma(\beta, x_i)\}_{i=1}^N$ followed by classical post-processing.
	\end{theorem}

	% \todo{Old version, of the theorem above is commented out here.}
	% \begin{theorem}[Learning algorithm for quantum Gibbs states]\label{learninggibbsalgo}
		% In the notation of the previous paragraph and assuming the decay of correlations of \Cref{decaycorrintro} uniformly in $x\in[-1,1]^m$, there exists an estimator $\hat{f}_O$ which depends on $x_1,\dots, x_N$ and such that, with probability $(1-\delta)$,
		% \begin{align*}
			%   \sup_{x\in[-1,1]^m}\, |f_O(x)-\hat{f}_{O}(x)|\le \epsilon\,\sum_{i=1}^M\|O_i\|_\infty\,,
			% \end{align*}
		% after being given $N=\Theta \left(\log\Big(\frac{M}{\delta}\Big)\,\log\Big(\frac{n}{\delta}\Big) e^{\operatorname{polylog}(\epsilon^{-1})} \right)$ samples. 
		% Moreover, the estimator $\hat{f}_O$ can be efficiently computed.
		% \end{theorem} 
	% \DSF{Add comment on scaling w.r.t. to locality}

	\noindent \textbf{Proof ideas}: Our estimator $\hat{f}_O$ is constructed as follows: during a training stage, we pick $N$ points $Y_1,\dots, Y_N\sim U$ and estimate the reduced Gibbs states over large enough enlargements $S_i\partial $ of the supports 
	$\mathcal{S}_i:=\{x_j|\,\operatorname{supp}(h_j(x_j))\cap S_i\partial\ne \emptyset\}\cap [x-\epsilon,x+\epsilon]^m$ of the observables $O_i$. Due to the anti-concentration property of the uniform distribution, the probability that a small region $\mathcal{S}_i\partial$ in parameter space contains $t$ variables $Y_{i_1},\dots, Y_{i_t}$ becomes large for $N\approx \log(M)$. We then run the classical shadow tomography protocol on those states in order to construct efficiently describable and computable product matrices $\widetilde{\sigma}(\beta, Y_1),\dots,\widetilde{\sigma}(\beta,Y_N)$. 
	Then for any region $S_i$, we select the shadows $\widetilde{\sigma}(\beta,Y_{i_1}),\dots \widetilde{\sigma}(\beta,Y_{i_t})$ whose local parameters are close to that of the target state and construct the empirical average $\widetilde{\sigma}_{S_i }(x):=\frac{1}{t}\sum_{j=1}^t\,\tr_{S_i^c}\big[\widetilde{\sigma}(\beta,Y_{i_j})\big]$.
	Using belief propagation methods (see \Cref{propgibbsshadow}), it is possible to show that exponential decay of correlations ensures that the estimator is a good approximation to local observables.
	Thus such operators can be well approximated using the reduced state $\tr_{S_i^c}\sigma(\beta,x)$ for $t\approx \log(n)$. 
	The estimator $\hat{f}_O$ is then naturally chosen as $\hat{f}_O(x):=\sum_{i=1}^M\,\tr[O_i\,\widetilde{\sigma}_{S_i }(x)]$. A key part of the proof is demonstrating that exponential decay of correlations implies that $f_O(x)$ does not change too much as $x$ varies.

	\subsubsection{Learning beyond exponentially decaying phases: ground states and indistinguishability}\label{sec:beyond_exp_decay}
	%So far we only discussed results for states for which there exist efficient classical algorithms to compute the expectation value of local observables. This follows directly from the fact that it is possible to compute local expectation values on some region $S$ by evaluating them on the localised Gibbs state $\sigma(\beta,x|_{S(r)})$, which can in turn be efficiently diagonalised. It would be desirable to extend our results to phases for which it is not known how to compute local expectation values efficiently on a classical computer. 
	
	%On the other hand, ground states and low-temperature Gibbs distributions will generally fail at satisfying such condition, and their local properties are moreover known to be hard to approximate in general \cite{Sly2010}. Thankfully, our previous strategy can be easily extended to include these non-trivial cases, by simply shifting the ``center'' of the parameters. More precisely, we introduce the following notion of \emph{generalised approximate local indistinguishability} (GALI)

	So far we have only discussed results for thermal states which have exponentially decaying correlations.
	It would be desirable to extend our results to phases for which this is not generally known to hold. We introduce a new condition called \emph{generalised approximate local indistinguishability} (GALI), under which learning local observables from samples can be done efficiently:
	
	\begin{definition}[Generalised approximate local indistinguishability (GALI)]\label{def:GALI_1}
		For $x_1^0,\ldots,x_m^0\in\mathbb{R}$ let 
		$\Phi\coloneqq \prod_{i=1}^m [-1+x_i^0,1+x_i^0]$ and for $x\in\Phi$ define $\rho(x)$ as either the ground state or thermal state of the local Hamiltonian $H(x)$. We say that the family of states $\rho(x)$ satisfies generalised approximate local indistinguishability (GALI) with parameter $\eta(S)\geq0$ and decay function $f$ if
		for any region region $S$ and $r\in\mathbb{N}$ there is a set of parameters $x^*_{\mathcal{S}(r)^c}$ s.t. for all $O$ supported on $S$ and $f_{O}(x):=\tr[O\,\rho(x)]$ the following bound holds: 
		\begin{align*}
			\sup_{x\in\Phi} |f_{O}(x)-f_{O}((x|_{\mathcal{S}(r)},x^*_{\mathcal{S}(r)^c})) | \le (|S|f(r)+\eta(S))\,\|O\|_\infty\,,
		\end{align*}
		for a function $f$ s.t. $\lim_{r\to\infty}f(r)=0$.
	\end{definition}
	GALI can be shown to hold for gapped ground state phases of matter and all thermal states with exponentially decaying correlations.
	We refer the reader to \Cref{sec:beyond_high_T} for further details. 
	Under the GALI assumption, we are able to prove the following generalisation of \Cref{learninggibbsalgo}:
	
	\begin{theorem}\label{thm:generalized_learning_gali}
		Let $\rho(x)$, $x\in \Phi\coloneqq \prod_{i=1}^m [-1+x_i^0,1+x_i^0]$, be a family of ground states or Gibbs states satisfying GALI, as per \cref{def:GALI}. 
		Given a set of $N$ samples $\{x_i,\tilde{\rho}( x_i)\}_{i=1}^N$, where $\tilde{\rho}(x_i)$ can be stored efficiently classically, and $N = O \big(\log\big(\frac{M}{\delta}\big)\,\log\big(\frac{n}{\delta}\big) e^{\operatorname{polylog}(\epsilon^{-1})} \big)$,
		there exists an algorithm that, on input $x\in\Phi$ and a local observable $O=
		\sum_{i=1}^M O_i$, produces an estimator $\hat{f}_O$
		such that, with probability $(1-\delta)$,
		\begin{align*}
			\sup_{x\in[-1,1]^m}\, |f_O(x)-\hat{f}_{O}(x)|\le \epsilon\,\sum_{i=1}^M\|O_i\|_\infty\,.
		\end{align*}
		Moreover, the samples $\tilde{\rho}(x_i)$ are efficiently generated from measurements of the Gibbs states $\{\rho( x_i)\}_{i=1}^N$ followed by classical post-processing.
	\end{theorem}
	\noindent Thus we can efficiently learn families of gapped ground states.
	We refer the reader to \cref{Sec:Learning_Ground_States} for proofs.

	\section{Comparison to previous work}\label{sec:previous_work}
	
	\subsection{Classical literature}

	The problem of Hamiltonian learning for classical models has attracted a lot of attention in the last years in the theoretical computer science community~\cite{Bresler2015,10.5555/3495724.3497094,Lokhov2018,pmlr-v119-zhang20l} which traditionally refers to it as Ising model --- or Markov field --- learning. 
	The question of what can be inferred from very few samples was also asked classically~\cite{one_or_multiple}. 
	Our work sheds further light on this question and is of interest even when restricting to classical systems. 
	Indeed, to the best of our knowledge, the statements of \Cref{cor:continuity_W1} and \Cref{localW1bound} are new even for  classical Gibbs distributions. 
	Previous work in~\cite{rouze2021learning} already established similar learning results for measures satisfying a so-called transportation cost inequality (TC)~\cite{Bobkov1999,Talagrand1996}, although the present condition of exponential decay of correlations is more standard.
	
	It should be noted that if a Gibbs measure satisfies TC, then any Lipschitz function of a random variable distributed according to it satisfies a Gaussian concentration bound~\cite{ledoux2001concentration}. This can easily be seen to imply that we can estimate the expectation value of $M$ Lipschitz functions up to an error $\epsilon$ with probability of success $\delta$ from $\mathcal{O}(\epsilon^{-2}\log(M\delta^{-1}))$ samples by taking the empirical average. At first sight this might look comparable with the sample complexity we obtain with our learning algorithm. However, this only holds for \emph{one} basis, whereas our result holds for any basis. Furthermore, if the number of Lipschitz observables satisfies $M=e^{\Omega(n)}$, then the number of samples required to obtain a good estimate through the empirical average becomes polynomial. On the other hand, given that $W_1(\sigma(\beta,x),\sigma(\beta,x'))\leq\epsilon n$, we can evaluate as many Lipschitz observables as we wish from $\sigma(\beta,x')$ without requiring any further samples. Thus, even for observables in a fixed basis our result has advantages.

	\subsection{Previous work on many-body quantum state tomography}

	As mentioned before, one striking advantage of our Gibbs tomography algorithm when estimating 
	expectation values of local observables compared to state-agnostic methods like classical shadows is the exponential speedup in the size of the support of the observable. In fact, our method gives good guarantees on the larger class of Lipschitz observables, which includes non-local observables. This advantage is even more visible when it comes to estimating entropic quantities: whereas the polynomial approximation proposed in \cite{Huang2020} works universally for any $n$-qubit state, it only gives good approximation guarantees for reduced states on very few qubits. Here instead, we avoid this issue by leveraging the Wasserstein continuity bounds offered in \cite{de2021quantum}. 
	%As a result, we get good estimates for reduced states on arbitrarily large regions of the lattice. 
	
	Our framework also differs from the one of Hamiltonian learning algorithms tackled in \cite{anshuweb,anshu2021sample,haah2021optimal}: in these papers, the authors were interested in estimating the parameter $x$ of a given Hamiltonian $H(x)$ given access to copies of the state $\sigma(\beta,x)$, in $\ell_2$ or $\ell_\infty$. Here instead, we argue that a good recovery in $W_1$ distance is implied by the weaker condition of recovery in $\ell_1$. Clearly, one can leverage these previous results to further control our $\ell_1$ bound, as we argue in \Cref{Gibbsstatetomo}. It should be noted however that our bound only requires that the Gibbs state $\sigma(\beta,x)$ satisfies an exponential decay of correlations, whereas these learning algorithms provide very efficient $\ell_\infty$ or $\ell_2$ recovery either for (i) commuting Hamiltonians or (ii) in the high-temperature regime. It remains an important question whether the condition of exponential decay of correlations is enough to get good $\ell_1$ recovery. Furthermore, in \Cref{sec:gibbs_learning_exponential_decay} we show that under the additional assumptions of uniform Markovianity and clustering of correlations, it is possible to learn in $W_1$ through the maximum entropy method, without resorting directly to learning the parameters $x$.

	\subsection{Previous work on learning observables in phases of matter} \label{Sec:Prev_Work_Phases}

	In \cite{huang2021provably}, the authors found a machine learning algorithm which, for any smoothly parameterised family of local Hamiltonians $\{H(x)\}_{x\in[-1,1]^m}$ in a finite spatial dimension with a constant spectral gap, can be trained to predict expected values of sums of local observables in the associated ground state $\psi_g(x)$. 
	More precisely, given a local observable $O=\sum_{i=1}^MO_i$ with $\operatorname{supp}(O_i)=\mathcal{O}(1)$, they construct an estimator $\hat{f}_O(x)$ of the expectation value of the observable such that 
	\begin{align}\label{Huangoldresult}
		\mathbb{E}_{x\sim U([-1,1]^m)}\,\Big[\big|\tr[O\psi_g(x)]-\hat{f}_{O}(x)\big|^2\Big]\le \epsilon^2 \Big(\sum_{i=1}^M \|O_i\|_\infty\Big)^2\,,
	\end{align}
	as long as the training size (i.e. the number of sampled points within the phase) is $N=\Big(\sum_{i=1}^M \|O_i\|_\infty\Big)^2\, m^{\mathcal{O}(1/\epsilon^2)}$.

	In \Cref{Sec:Learning_Ground_States}, we improve this
	result for ground states in three ways, up to further imposing the GALI condition: first, we can assume that the parameters $x$ are distributed according to a much larger class of distributions than the uniform distribution. 
	This extension does not carry so easily in the proof of \cite{huang2021provably} which uses Fourier analysis techniques involving integration over the Lebesgue measure to derive \Cref{Huangoldresult}.
	Second, theirs is a result in expectation, that is in $\|.\|_{L^2}$, whereas our bound in \Cref{thm:generalized_learning_gali} works in the worst-case setting associated to the stronger $\|.\|_\infty$-norm topology. 
	Third and most importantly, the dependence of the number of training data points scales exponentially in the precision parameter $\epsilon$ in \Cref{Huangoldresult}, whereas ours scales only quasi-polynomially. 
	
	Finally, we extend the learning result beyond ground states to finite temperature phases of matter with exponential decay of correlations.
	This not only includes all high-temperature phases of matter (regardless of the Hamiltonian), but also low-temperature phases with the relevant correlation functions \cite{DuminilCopin2019}.
	This is a particularly relevant result since zero temperature is never achieved in practice, so in reality we are always working with low-temperature thermal states.

	We also recognise independent, concurrent work by \cite{Lewis_Huang_Preskill_2023}.
	Here the authors consider the same setup of gapped ground states as \cite{huang2021provably} and also improved over \Cref{Huangoldresult} to achieve the same sample complexity as \Cref{thm:generalized_learning_gali}. 
	However, their result is not directly comparable to ours.
	We emphasise \cite{Lewis_Huang_Preskill_2023} consider gapped, ground state phases, whereas our work includes thermal phases.
	We also note they remove all conditions on the prior distribution over the samples $x$, whereas we still need to assume a type of mild anti-concentration over the local marginals. 
	However, their result is still stated as an $\|.\|_{L^2}$-bound due to the use of machine learning techniques, whereas our more straightforward approximation tools allow us to get stronger bounds in $\|.\|_\infty$. Conceptually speaking, our methods for approximating local expectation values requires no knowledge of machine learning techniques. 
	Our work also shows that it is possible to go beyond gapped quantum phases and learn thermal phases satisfying GALI.
	\begin{comment}
		
		\textcolor{blue}{CR: we might want to distinguish ourselves a bit more} Finally, in \Cref{Sec:ML_Thermal_Matter} we demonstrate that the machine learning techniques from \cite{Lewis_Huang_Preskill_2023} can be extended from quantum phases corresponding to gapped ground states, to thermal phases of matter with exponentially decaying correlations.
	\end{comment}
	
	\section{Discussion and Conclusions}
	In this paper we 
	%extended and improved two recent results on the 
	contributed to the tasks of tomography and learnability of quantum many-body states by combining previous techniques with approaches not considered so far in this field, in order to obtain novel and powerful features.
	
	\medskip

	\textbf{Tomography.} \ 
	First, we extended the results of~\cite{rouze2021learning} on the efficient tomography of high-temperature commuting Gibbs states
	% \todo{State the results without refering to a previous paper that people may not have read.}
	to Gibbs states with exponentially decaying correlations. 
	This result permits to significantly enlarge the class of states for which we know how to learn all quasi-local properties with a number of samples that scales polylogarithmically with the system's size. In particular, our results now also hold for classes of Gibbs states of non-commuting Hamiltonians. 
	As we require exponentially fewer samples to learn in the Wasserstein metric when compared with the usual trace distance and still recover essentially all physically relevant quantities associated to the states, we hope that our results motivate the community to consider various tomography problems in the Wasserstein instead of trace distance.  
	
	As we achieved this result by reducing the problem of learning the states to learning the parameters of the Hamiltonian in $\ell_1$, we hope our work further motivates the study of the Hamiltonian learning problem in $\ell_1$-norm with polylog samples. $1$D Gibbs states are a natural place to start, but obtaining Hamiltonian learning algorithms just departing from exponential decay of correlations would provide us with a complete picture. In \Cref{sec:gibbs_learning_exponential_decay} we also partially decoupled the Hamiltonian learning problem from the $W_1$ learning one by resorting to the uniform Markov condition. Thus, it would be important to establish the latter for a larger number of systems.
	
	It would be interesting to investigate the sharpness of our bounds, and to understand if exponential decay of correlations is really necessary. One way of settling this question would be to prove polynomial lower bounds for learning in Wasserstein distance for states at critical temperatures. 
	
	\medskip

	\textbf{Learning Phases of Matter.} 
	\ Second, we improved the results of~\cite{huang2021provably} for learning a class of states in several directions, including the scaling in precision, the classes of states it applies to and the form of the recovery guarantee. 
	In particular, the results now apply to Gibbs states, which are the states of matter commonly encountered experimentally.
	Interestingly, we did not need to resort to machine learning techniques to achieve an exponentially better scaling in precision by making arguably mild assumptions on the distributions the states are drawn from. 
	%  \textcolor{blue}{CR: What exactly do we hope to achieve?} Thus, it would be interesting to see if combining our results with machine learning techniques it would be possible to further improve the scaling in the precision and other parts of the algorithm.
	Although the results proved here push the state-of-the-art of learning quantum states, we believe that our methods, for instance the novel continuity bounds for various local properties of quantum many-body states, will find applications in other areas of quantum information.
	
	Beyond the GALI thermal phases and ground states studied here, it would be interesting to find other families of states which can be efficiently learned, and indeed if more restrictive assumptions on the parameterisation of Hamiltonians can result in more efficient learning. One interesting open problem that goes beyond the present paper's scope is finding families of states satisfying GALI without belonging to a common gapped phase of matter. If such a family existed, it would clarify the differences between our framework and that of \cite{Lewis_Huang_Preskill_2023}.
	Finally, we realise that although the results proved here are for lattice systems, they almost certainly generalise to non-lattice configurations of particles.
	
	\begin{comment}
		
		\textcolor{blue}{three open problems: $l_1$ algo instead of $l_\infty$ algo under weaker decay of correlations in the Gibbs tomography result, ground states poly region for sums of local beyond our current LTQO/commuting setting, Wasserstein extension of the Gibbs learning}
		
		\DSF{another cool application of Hamiltonian learning. It would be cool to extend the results of Hamiltonian to more systems. Clear example: 1D.}

		\textcolor{red}{Potential things to talk about:
			\begin{itemize}
				\item New learning algorithm without ML and achieves same scaling to the ML algorithm.
				\item Might suggest than an ML algorithm should be able to achieve better scaling?
				\item Alternatively this could be an alternative to Huang's ML work for Gibbs and Ground states.
			\end{itemize}
		}
	\end{comment}
	\clearpage 
	\section{Acknowledgments}
	
	The authors gratefully recognise useful discussions with
	\begingroup
	\hypersetup{urlcolor=navyblue}
	\href{https://orcid.org/0000-0001-5317-2613}{Hsin-Yuan (Robert) Huang} and \href{https://orcid.org/0000-0001-9537-9663}{Haonan Zhang}.
	We thank Laura Lewis,  Viet T. Tran, Sebastian Lehner, 
	\href{https://orcid.org/0000-0002-8291-648X}{Richard Kueng}, 
	\href{https://orcid.org/0000-0001-5317-2613}{Hsin-Yuan (Robert) Huang}, and 
	\href{https://inspirehep.net/authors/992791}{John Preskill} for sharing a preliminary of the manuscript \cite{Lewis_Huang_Preskill_2023}, which was discussed in \Cref{Sec:Prev_Work_Phases}.
	\endgroup
	
	EO is supported by the Munich Quantum Valley and the Bavarian state government, with funds from the Hightech Agenda Bayern Plus.
	CR acknowledges financial support from a Junior Researcher START Fellowship from the DFG cluster of excellence 2111 (Munich Center for Quantum Science and Technology), from the ANR project QTraj (ANR-20-CE40-0024-01) of the French National Research Agency (ANR), as well as from the Humboldt Foundation.
	DSF is supported by France 2030 under the French National Research Agency award number “ANR-22-PNCQ-0002”.
	JDW acknowledges support from the United States Department of Energy, Office of Science, Office of Advanced Scientific Computing Research, Accelerated Research in Quantum Computing program, and also NSF QLCI grant OMA-2120757.
	
	\bibliography{References}
	\bibliographystyle{alpha}

	\clearpage
	\appendix

	\section*{Supplemental Material}

	\section{Preliminaries}
	
	Given a finite dimensional Hilbert space $\mathcal{H}$,
	% and $\mathcal{H}$ be finite dimensional Hilbert spaces.
	we denote by $\mathcal{B}(\mathcal{H})$ 
	% and $\mathcal{B}(\mathcal{H})$ denote 
	{the algebra of bounded operators on $\mathcal{H}$}, whereas $\mathcal{B}_{\operatorname{sa}}(\cH)$ denotes the subspace of self-adjoint operators.
	% and $\mathcal{H}$, respectively. 
	We denote by $\mathcal{D}(\mathcal{H})$ the set of positive operators on $\mathcal{H}$ of unit trace, and by $\mathcal{D}_{+}(\mathcal{H})$ the subset of positive, full-rank operators on $\mathcal{H}$. 
	Schatten norms are denoted by $\|.\|_p$ for $p\ge 1$. The identity matrix in $\mathcal{B}(\mathcal{H})$ is denoted by $I$. 
	Given a bipartite system $AB$, the normalised partial trace over a subsystem $A$ is written $\tau_A$, i.e. $\tau_A:=2^{-|A|}\tr_A$.

	In this work, we consider a family of local qubit interactions $\{h_j(x_j)\}_{x_j\in[-1,1]^\ell}$, $j=1,\dots ,n$ over the $D$-dimensional lattice $\Lambda=[-L,L]^D$, for some fixed integer $\ell$, where $n=(2L+1)^D$ denotes the total number of qubits constituting the system. For each $j$ and all $x_j\in [-1,1]^\ell$, $h_j(x_j)$ is supported on a ball $A_j$ around site $j\in\Lambda $ of radius $r_0$. We also assume that the matrix-valued functions $x_j\mapsto h_j(x_j)$ as well as their derivatives are uniformly bounded: $\|h_j\|_\infty,\|\nabla_x h_j(x)\|_\infty\le h$. 
	For sake of simplicity, we assume that the interactions are linear functions of their parameters, that is $h_j(x_j) = x_jV_j$ for some fixed operator $V_j$. 
	However this assumption is not necessary in any of our proofs, as commented in \Cref{nonlinearparam}. 
	Concatenating the vectors $x_j$ into $x=(x_1,\dots, x_n)=(x'_1,\dots,x'_m)$, $m=n\ell$, the local interactions induce the following family of Hamiltonians $\{H(x)\}_{x\in [-1,1]^{m}}$, with:
	\begin{align}\label{Hamiltoniandef}
		H(x)=\sum_{j=1}^m h_j(x_j)\,.
	\end{align}
	% Here we attack the problem for Gibbs states. We take a family of Hamiltonians $\{H(x)\}_{x\in \Omega}$ on some $D$-dimensional lattice $\Lambda$, for some open and bounded subset $\Omega\subset \mathbb{R}^m$, assume that $H(x)=\sum_{j=1}^{\tilde{n}} h_j(x)$ \DSF{smooth function?}\JDW{to DSF: Isn't the bound on the derivative below sufficient?} where each of the terms $h_j$ acts on $|A_j|=\mathcal{O}(1)$ number of sites in a ball of radius $\mathcal{O}(1)$. For sake of simplicity, we assume that $\tilde{n}$ scales linearly with $n$. We also assume that each term $\|h_j(x)\|_\infty\le h$ uniformly in $j$ and $x$, and that $\|\partial_{\hat{u}}h_j(x)\|_\infty\le \ell$ for any unit tangent vector $\hat{u}$. 
	More generally, given a region $B\subset \Lambda$ of the lattice, we denote by $H_B(x):=\sum_{j|A_j\subset B}h_j(x)$ the Hamiltonian restricted to $B$. 
	We denote by $x|_{\mathcal{S}(r)}$ the concatenation of vectors $x_j$ corresponding to interactions $h_j$ supported on regions intersecting $S(r):=\{l\in\Lambda|\,\operatorname{dist}(l,S)\le r\}$.
	Given two operators $A,B$, then $\operatorname{dist}(A,B)$ denotes the distance between $\operatorname{supp}(A)$ and $\operatorname{supp}(B)$. For much of the following, we will be concerned with Gibbs states, defined as 
	\begin{align*}
		\sigma(\beta,x):=\frac{e^{-\beta H(x)}}{\tr[e^{-\beta H(x)}]}.
	\end{align*}
	Note that with the parametrisation above the maximally mixed is always contained in the image of $\sigma(x,\beta)$ by taking $\sigma(0,\beta)$. For some of our results we will want to make sure that we are not taking a parametrisation that also has states in the trivial in its image. In that case, we will write 
	
	\begin{align}\label{equ:shifted_Hamiltonian}
		\sigma(\beta,x,H_0):=\frac{e^{-\beta H(x)-H_0}}{\tr[e^{-\beta H(x)-H_0}]},
	\end{align}
	where $H_0$ is another Hamiltonian that potentially shfits us away from the trivial phase.
	
	In particular, we will be interested in systems satisfying the following type of correlation decay: 
	\begin{condition}[Exponential Decay of Correlations] \label{decaycorr}
		For a state $\sigma$ and any operator $X_A$, resp.~$X_B$, supported on region $A$, resp. $B$, we say the state satisfies \emph{exponential decay of correlations} if
		\begin{align}
			\operatorname{Cov}_{\sigma}(X_A,X_B)\le C\,\min\{|A|,\,|B|\}\,\|X_A\|_\infty\,\|X_B\|_\infty\,e^{-\nu \operatorname{dist}(A,B)}\,,
		\end{align}
		for any choice of $X_A$,$X_B$, and
		for some parameters $C,\nu>0$ which we assume independent of $x$ and of the lattice size $n$, and where 
		\begin{align*}
			\operatorname{Cov}_\sigma(A,B):=\frac{1}{2}\,\tr\Big[\sigma\,\big\{A-\tr[\sigma A],B-\tr[\sigma B]\big\}\Big]\,.
		\end{align*} 
	\end{condition}
	
	\Cref{decaycorr} is satisfied by many classes of Gibbs states, including high-temperature Gibbs states~\cite{harrow2020classical,kuwahara_gibbs} and $1$D Gibbs states at any constant temperature~\cite{harrow2020classical,Bluhm2022exponentialdecayof}. It is also known to hold for ground states of gapped Hamiltonians~\cite{Hastings2006}. In fact, the class of Gibbs states for which \Cref{decaycorr} holds is larger than that for which polylog algorithms to learn the parameters of the Hamiltonian are known. In \Cref{sec:learning_gibbs_examples} we will discuss several examples for which it is known how to learn the parameters efficiently. In \Cref{sec:gibbs_learning_exponential_decay} we will also consider the case when we have the additional assumption of uniform Markovianity to show that then it is possible to bypass having to learn the parameters.

	\section{Lipschitz observables}\label{AppendixLip}
	
	In this appendix, we argue that Lipschitz observables and the induced Wasserstein distance capture most observables of physical interest, such as local and quasi-local observables, as well as quasi-local polynomials of the state and entropic quantities of subsystems. They can even capture global properties, including some of physical interest like global entropies. These classes of examples justify the claim that Lipschitz observables and the Wasserstein distance capture well both linear and nonlinear extensive properties of quantum states. 
	
	Let us illustrate our previous claims. An important class of Lipschitz observables are those of the form
	\begin{align}\label{equ:local_average}
		\sum_{i=1}^MO_i, \quad M=\mathcal{O}(n),\quad \|O_i\|=\mathcal{O}(1),\quad \max\limits_{1\leq j\leq n}|\{i:\operatorname{supp}(O_i)\cap\{j\}\not= \emptyset\}|=\mathcal{O}(1).
	\end{align}
	
	Observables like those defined in \Cref{equ:local_average} include local observables w.r.t. to a regular lattice.
	However, it is also not difficult to see that the expectation values of such observables are characterised by the marginals of the states on a few qubits. But Lipschitz observables capture more than strictly local properties. Indeed, as shown in~\cite{rouze2021learning}, the time evolution of local observables like those in~\Cref{equ:local_average} by a shallow quantum circuit or a short continuous-time evolution satisfying a Lieb-Robinson bound are Lipschitz. These include evolutions by Hamiltonians with algebraically decaying interactions, which will map strictly local Hamiltonians to quasi-local observables. In fact, recent results~\cite{randomlip} show that Lipschitz observables can distinguish two random quantum states almost optimally. As such states are locally indistinguishable~\cite[Corollary 15]{Brando2016}, this fact shows that Lipschitz observables capture much more than just quasi-local properties of quantum states.

	Although so far we only discussed how to use the Wasserstein distance to control linear functionals of the state, the fact that the Wasserstein distance behaves well under tensor products means that it is also easy to control the error for non-linear functions. Indeed, in~\cite[Propostion 4]{de2021quantum}, the authors show that the Wasserstein distance is additive under tensor products. i.e. for all states $\rho,\sigma$ and integer $k$ we have
	\begin{align}
		W_1(\rho^{\otimes k},\sigma^{\otimes k})=kW_1(\rho,\sigma).
	\end{align}
	We can then combine this additivity with the standard trick that a polynomial of degree $k$ on a quantum state can be expressed as the expectation value of a certain observable $O$ on $\rho^{\otimes k}$. In particular, if this polynomial is an average over polynomials in reduced density matrices of constant size, it is not difficult to see that the corresponding observable on $\rho^{\otimes k}$ will be Lipschitz as well.
	
	Let us exemplify this in the case of the average purity of a state. For a subset $A\subset[n]$ of the qubits of size $l$, let $\mathbb{F}_A\in\left(\mathbb{C}^{2}\right)^{\otimes 2l}$ be the flip operator acting on two copies of those qubits:
	\begin{align}
		\mathbb{F}_A(\ket{\psi}\otimes \ket{\varphi})=\ket{\varphi}\otimes \ket{\psi}.
	\end{align}
	It can be shown in a few lines that $\tr\left[\mathbb{F}_A\rho^{\otimes 2}\right]=\tr\left[\rho_A^2\right]$. Furthermore, observables of the form
	\begin{align}
		O=\sum_{i=1}^M\mathbb{F}_{A_i},\quad M=\mathcal{O}(n),\quad \max\limits_{1\leq j\leq n}|\{i:A_i\cap\{j\}\not= \emptyset\}|=\mathcal{O}(1).
	\end{align}
	satisfy $\|O\|_{\operatorname{Lip}}=\mathcal{O}(1)$. Then
	\begin{align}
		\sum_{i=1}^M\tr\left[\rho_{A_i}^2-\sigma_{A_i}^2\right]=\tr\left[O(\rho^{\otimes 2}-\sigma^{\otimes2})\right]\leq \|O\|_{\operatorname{Lip}}W_1(\rho^{\otimes 2},\sigma^{\otimes 2})=2\|O\|_{\operatorname{Lip}}W_1(\rho,\sigma).
	\end{align}
	By a direct generalisation of the above, we see that $W_1(\rho,\sigma)=\mathcal{O}(\epsilon n/k)$ is sufficient to ensure that degree $k$ polynomials of the states are approximated to a multiplicative error. As we will see later in \Cref{beyondlin}, this polynomial trick can be used to ensure that averages of various subsystem entropies, mutual informations and conditional mutual informations are well-approximated given a Wasserstein bound.
	
	Once again it should be emphasised that a Wasserstein bound can be used to control global properties, even non-linear ones. A good example of that is the entropy of a quantum state. In~\cite[Theorem 1]{de2021quantum}, the authors show the continuity bound:
	\begin{align}\label{eqentropycontbound}
		|S(\rho)-S(\sigma)|\leq g(W_1(\rho,\sigma))+W_1(\rho,\sigma)\log(4n),
	\end{align}
	where $g(t)=(t+1)\log(t+1)-t\log(t)$. In this case, it turns out that a Wasserstein distance of $W_1(\rho,\sigma)=\mathcal{O}(\epsilon n/\log(n))$ suffices to obtain a multiplicative error for the entropy.
	Finally, it is also worth mentioning observables that are not Lipschitz. Simple examples include linear combinations of high-weight Paulis.

	\section{Gibbs states tomography}\label{sectomogibbs}
	
	In this section, our main goal is to devise an efficient tomography algorithm for Gibbs states $\sigma(\beta,x)$.
	In particular, we wish to learn the parameters $x$ to high precision.
	We prove the following lemma:
	% \todo{JDW: Here I've just restated the main theorem from the tomography section.}
	
	\begin{theorem}[Tomography algorithm for decaying Gibbs states ]\label{maintheorem_full}
		Let $H(x)= \sum_i h_i(x_i)$ be a Hamiltonian such that each $h_i(x_i), \ x_i\in [-1,1]^\ell$, is not more than $k$-local, for $k=O(1)$, and all terms commute.  
		For some unknown $x$, let $\sigma(\beta,x)$ be its associated Gibbs state satisfying exponential decay of correlations as per \Cref{decaycorr}.
		Then there exists an algorithm that provides the description of parameters $x'$ such that the state $\sigma(\beta,x')$ satisfies:
		\begin{align}\label{equ:target_W1}
			W_1(\sigma(\beta,x),\sigma(\beta,x'))\leq \epsilon n  
		\end{align}
		with probability greater than $1-\delta,$ such that the algorithm requires access to no more than $N=\mathcal{O}\big(\log(\delta^{-1})\operatorname{polylog}(n)\,\epsilon^{-2}\big)$ samples of the state (see \Cref{seccommuting}). 
		
		The result extends to the case where  $\{h_i(x_i)\}_i$ do not commute whenever one of the following two assumptions is satisfied: 
		\begin{enumerate}[(i)]
			\item  the high-temperature regime, $\beta<\beta_c$ (see \Cref{sechightemperature}).
			\item uniform clustering/Markov conditions (see \Cref{coroapproxmarkov}).
		\end{enumerate}
		In case (ii), we find good approximation guarantees under the following slightly worst scaling in the precision $\epsilon$: $N=\mathcal{O}(\epsilon^{-4}\operatorname{polylog}(n\delta^{-1}))$.
	\end{theorem}
	\begin{proof}[Proof Outline]
		The full proof is laid out in sections \ref{Sec:Quantum_Belief_Propagation}, \ref{Sec:Continuity_Estimates} and \ref{sec:learning_gibbs_examples}.
		
		The fundamental part of the result uses the continuity estimate of the Wasserstein distance between two Gibbs states that is of interest on its own. 
		In \Cref{cor:continuity_W1} we will show that under exponential decay of correlations we have:
		\begin{align}\label{upperboundW1_first}
			W_1(\sigma(\beta,x),\sigma(\beta,y))\le  \,\|x-y\|_{\ell_1}\,\operatorname{polylog}(n)\,.
		\end{align}
		The significance of the bound in \Cref{upperboundW1_first} is that it reduces the problem of obtaining a good estimate of $\sigma(\beta,x)$ in $W_1$ to estimating the parameters $x$ in $\ell_1$ distance. 
		This is a variation of the Hamiltonian learning problem~\cite{anshu2021sample,coles_learning,haah2021optimal}, and we can then directly import results from the literature for our tomography algorithm.
	\end{proof}

	%In this section, our main goal is to devise an efficient tomography algorithm for Gibbs states $\sigma(\beta,x)$. 
	%More precisely, we are interested in tomography algorithms that output with probability of success at least $1-\delta$, and given $\operatorname{poly}(\epsilon^{-1},\log(n),\log(\delta^{-1}))$ copies of a Gibbs state $\sigma(\beta,x)$, a Gibbs state $\sigma(\beta,x')$ satisfying
	%\begin{align}\label{equ:target_W1}
	%   W_1(\sigma(\beta,x),\sigma(\beta,x'))\leq \epsilon n.
	%\end{align}

	As we argued before in \Cref{Lipsec}, the recovery guarantee in \Cref{equ:target_W1} suffices to ensure that $\sigma(\beta,x')$ mirrors all the quasi-local properties of $\sigma(\beta,x)$. 
	Furthermore, the polylog complexity in system size is exponentially better than what is required to obtain a recovery guarantee in trace distance~\cite[Appendix G]{rouze2021learning}, even for product states.

	\subsection{Quantum belief propagation} \label{Sec:Quantum_Belief_Propagation}

	We start by recalling a well-known tool in the analysis of quantum Gibbs states known as quantum belief propagation \cite{hastings2007quantum,kim2017markovian,kato2019quantum}. 
	We assume a parameterisation of the Hamiltonian as $H(x) = \sum_{j=1}^m x_j V_j$ for appropriate operators $V_j$ (we will generalise this to other parameterisations later) and for some observable $L$ we define the function $f_L(\beta,x)=\tr\left[\sigma(\beta,x)L\right]$.
	The belief propagation method then states that we have that for any $k\in[m]$,
	\begin{align*}
		\partial_{x'_k}f_L(\beta,x)=-\frac{\beta}{2}\,\tr\Big[L\big\{\Phi_{H(x)}(\partial_{x'_k} H(x)),\sigma(\beta,x)\big\}\Big]+\beta \tr(\partial_{x'_k}H(x)\sigma(\beta,x))\,\tr(L \sigma(
		\beta,x))\,.
	\end{align*}
	where the quantum belief propagation operator $\Phi_{H(x)}$ is defined as
	\begin{align*}
		\Phi_{H(x)}(V):=\int_{-\infty}^\infty\,dt\,\kappa_\beta(t)\,e^{-iH(x)t}Ve^{iH(x)t}\,, %\quad \text{for} \quad H(s) := H_0 + sV,
	\end{align*}
	for some smooth, fast-decaying probability density function $\kappa_\beta(t):=\frac{1}{2\pi}\int \widetilde{\kappa}_\beta(\omega)e^{i\omega t}d\omega$ of Fourier transform
	\begin{align*}
		\widetilde{\kappa}_\beta(\omega):=\frac{\tanh(\beta\omega/2)}{\beta\omega/2}\,.
	\end{align*}
	The function $\kappa_\beta$ was in fact computed in \cite[Appendix B]{anshu2021sample}: for $t\in\mathbb{R}\backslash \{0\}$:
	\begin{align}\label{fbetabound}
		\kappa_\beta(t):=\frac{2}{\pi\beta}\,\log\frac{e^{\pi |t|/\beta}+1}{e^{\pi|t|/\beta}-1}\le  \frac{4}{\pi\beta}\, \frac{1}{e^{\pi |t|/\beta}-1}
	\end{align}
	Rewriting the above derivative, and using the notations $\langle O\rangle_{\beta,x} \equiv \tr(\sigma(\beta,x)O)$ for the expected value of an observable $O$ in the Gibbs state $\sigma(\beta,x)$, we have that 
	\begin{align} \label{Eq:Partial_f_L}
		\partial_{x'_k}f_L(\beta,x)=-\frac{\beta}{2}\,\langle \big\{L,\,\widetilde{H}_{k}(x)-\langle \widetilde{H}_{k}(x)\rangle_{\beta,x}\big\}  \rangle_{\beta,x}
	\end{align}
	where $\widetilde{H}_{k}(x):=\Phi_{H(x)}(\partial_{x'_k}H(x))$. We define the covariance between two observables $A$ and $B$ in the state $
	\sigma$ as 
	\begin{align*}
		\operatorname{Cov}_\sigma(A,B):=\frac{1}{2}\,\tr\Big[\sigma\,\big\{A-\tr[\sigma A],B-\tr[\sigma B]\big\}\Big]\,.
	\end{align*}
	Therefore 
	\begin{align}\label{derivativefL}
		\partial_{x'_k}f_L(\beta,x)=-\beta\operatorname{Cov}_{\sigma(\beta,x)}\,(L,\widetilde{H}_{k}(x))\,.
	\end{align}
	In what follows, we will need to approximate $\widetilde{H}_{k}(x)$ by observables supported on bounded regions. For this, we make use of Lieb-Robinson bounds for Hamiltonians of finite-range interactions \cite{lieb1972finite,poulin2010lieb,kliesch2014lieb,PhysRevLett.97.050401,hastings2010locality,sidoravicius2009new,cubitt2015stability}. Here we choose a version proven in \cite[Lemma 5.5]{cubitt2015stability}: for any observable $O_A$ supported on a region $A$ of the lattice, and any $B\supset A$, we denote by $\alpha_t$, resp. by $\alpha^B_t$, the unitary evolution generated by $H(x)$, resp. by $H_B(x)$, up to time $t$, i.e.
	\begin{align*}
		\alpha_t(O):=e^{-iH(x)t}Oe^{iH(x)t}\,,\qquad  \alpha^B_t(O):=e^{-iH_B(x)t}Oe^{iH_B(x)t}\,.
	\end{align*}
	which then satisfy
	\begin{align}\label{LRbound}
		\|\alpha_t(O_A)-\alpha_t^B(O_A)\|_\infty \le c\, |A|\,\|O_A\|_\infty \,e^{vt-\mu\,\operatorname{dist}(A,B^c)}\,,
	\end{align}
	for some parameters $c,v,\mu>0$ which depend on the interactions $h_j$ but can be chosen independent of $n$ and $x$.

	\begin{lemma}\label{LemmaapproxPhi}
		For any region $A\subset B\subset \Lambda$ and operator $O_A$ supported in $A$ and all $x$,
		\begin{align*}
			\|\Phi_{H(x)}(O_A)-\Phi_{H_B(x)}(O_A)\|_\infty \le c'\,|A| \,\|O_A\|_\infty\,   e^{-\mu' \operatorname{dist}(A,B^c)}
		\end{align*}
		for some parameters $c'$ and $\mu'$ depending on $H(x)$ and $\beta$ but independent of $n$.
	\end{lemma}
	\begin{proof}
		We make use of the exponential decay of $\kappa_\beta$ provided in \Cref{fbetabound} together with the Lieb-Robinson bound \Cref{LRbound}:  
		\begin{align*}
			\|\Phi_{H(x)}(O_A)-\Phi_{H_B(x)}(O_A)\|_\infty &\le \int_{-\infty}^\infty |\kappa_\beta(t)|\,\|\alpha_t(O_A)-\alpha_t^B(O_A)\|_\infty\,dt\\
			&\le c\, |A|\,\|O_A\|_\infty  e^{-\mu \operatorname{dist}(A,B^c)} \int_{-\delta}^{\delta}\,|\kappa_\beta(t)| \,\,e^{vt}\,dt\\
			&\quad +2 \,\|O_A\|_\infty\, \int_{[-\delta,\delta]^c}\,|\kappa_\beta(t)|\,dt\,.
		\end{align*}
		For the first integral above, we use that $|\kappa_\beta(t)|\propto \log(1/t)$ for $t$ small. More precisely,
		\begin{align*}
			\int_{-\delta}^\delta |\kappa_\beta(t)|\,e^{vt}\,dt&\le \frac{4e^{v\delta}}{\pi\beta}\,\int_{0}^\delta\, \log\left(\frac{e^{\pi t/\beta}+1}{t\pi/\beta}\right)\,dt\le \frac{4e^{(v+\pi/\beta)\delta}}{\pi^2}
		\end{align*}
		For the other integral, we use the exponential decay of $\kappa_\beta$:
		\begin{align*}
			\int_{[-\delta,\delta]^c}\,|\kappa_\beta(t)|\,dt\le\frac{8}{\pi\beta} \int_{\delta}^\infty\,\frac{1}{e^{\pi t/\beta}-1}\,dt\le \frac{8}{\pi\beta}\int_\delta^\infty e^{-\frac{\pi t}{2\beta}}dt=\frac{16}{\pi^2}\,e^{-\frac{\pi\delta}{2\beta}}\,,
		\end{align*}
		where the second inequality holds for $\delta\ge \frac{2\beta}{\pi}\operatorname{sh}^{-1}\big(\frac{1}{2}\big)\equiv \delta_1$.
		Choosing $\delta:=\delta_1+\mu \operatorname{dist}(A,B^c)/(2\big(v+\pi/\beta\big))$, we get
		\begin{align*}
			\|\Phi_{H(x)}(O_A)-\Phi_{H_B(x)}(O_A)\|_\infty \le c'\,|A|\|O_A\|_\infty\,e^{-\mu'\operatorname{dist}(A,B^c)}\,,
		\end{align*}
		for some constant $c'\equiv c'(\beta,v)$, where $\mu'=\mu\min\big\{\frac{1}{2},\,\frac{\pi}{4\beta(v+\pi/\beta)} \big\}$.
	\end{proof}

	\subsection{Continuity estimate for $W_1$ distance on Gibbs states}
	\label{Sec:Continuity_Estimates}
	In this subsection, we will prove \Cref{upperboundW1_first}.
	First, we use the bound derived in \Cref{LemmaapproxPhi} together with the assumption that $\sigma(\beta,x)$ has exponential decay of correlations in order to control the derivatives $\partial_{x'_k}f_L$:
	\begin{proposition}\label{prop:bounds}
		Assume that $\sigma(\beta,x)$ satisfies the condition of decay of correlations, \Cref{decaycorr}. 
		Then for any $k\in[m]$,
		\begin{align}\label{eqpartiali}
			|  \partial_{x'_k}f_L(\beta,x)|\le \,\|L\|_{\operatorname{Lip}}\,\operatorname{polylog}(n)\,,
		\end{align}
		for some polynomial of $\log(n)$ of degree $D$ with coefficients depending on $\beta,r_0,D,h,c',\nu,\mu'$ and $C$.
	\end{proposition}
	
	\begin{proof}
		Denoting by $j_k$ the index of the interaction $h_{j_k}$ which depends on variable $x_k'$, we have that, given $\Phi_{H(x)}(\partial_{x'_k}h_j)=\delta_{j,j_k}\Phi_{H(x)}(\partial_{x'_k}h_{j_k})$, and denoting $\widetilde{h}_{k}=\Phi_{H(x)}(\partial_{x'_k}h_{j_k})$, from \Cref{derivativefL} we have:
		\begin{align*}
			|  \partial_{x'_k}f_L(\beta,x)| ={\beta}\,\operatorname{Cov}_{\sigma(\beta,x)}(L,\widetilde{H}_k(x))=\beta\,\operatorname{Cov}_{\sigma(\beta,x)}(L,\,\widetilde{h}_{k})\,.
		\end{align*}
		Next, given a region $B_k\supset A_{j_k}$, define the observable 
		\begin{align}\label{obsOBl}
			O_{B_k}:=\Phi_{H_{B_k}(x)}(\partial_{x'_k}h_{j_k})-\langle\Phi_{H_{B_k}(x)}(\partial_{x'_k}h_{j_k})\rangle_{\beta, x}\, .
		\end{align}
		Then by \Cref{LemmaapproxPhi} we have that
		\begin{align*}
			\operatorname{Cov}_{\sigma(\beta,x)}(L,\,\widetilde{h}_{k}(x)) &=
			\operatorname{Cov}_{\sigma(\beta,x)}(L,\,\widetilde{h}_{k}(x)-O_{B_k})+\operatorname{Cov}_{\sigma(\beta,x)}(L,\,O_{B_k})\\ 
			&\le 2 \|L\|_{\infty}\, \|\Phi_{H(x)}(\partial_{x'_k}h_{j_k})-\Phi_{H_{B_k}(x)}(\partial_{x'_k}h_{j_k})\|_\infty+ \operatorname{Cov}_{\sigma(\beta,x)}(L,\,O_{B_k})\\
			&\le  2n c'(2r_0)^D\,h\,\|L\|_{\operatorname{Lip}}\,e^{-\mu'\operatorname{dist}(A_{j_k},B_k^c)}+\operatorname{Cov}_{\sigma(\beta,x)}(L,O_{B_k})\,.
		\end{align*}
		Next, we estimate the last covariance above. Denoting $B_k(r):=\{i\in\Lambda:\,\operatorname{dist}(i,B_k)\le r\}$, we get
		\begin{align*}
			\operatorname{Cov}_{\sigma(\beta,x)}(L,O_{B_k})&=\operatorname{Cov}_{\sigma(\beta,x)}(L-\tau_{B_k(r)}(L),O_{B_k})+\operatorname{Cov}_{\sigma(\beta,x)}(\tau_{B_k(r)}(L),O_{B_k})\\
			&\le 2h \|L-\tau_{B_k(r)}(L)\|_\infty+ 2C|B_k|\,h \|L\|_\infty\,e^{-\nu r} \\
			&\le 2h |B_k(r)|\,\|L\|_{\operatorname{Lip}}+ 2C|B_k|\,h\,n \|L\|_{\operatorname{Lip}}\,e^{-\nu r}\,,
		\end{align*}
		where the second line above follows from the condition of decay of correlations \Cref{decaycorr}. 
		Choosing $B_k=A_{j_k}(\lfloor \log({n})/\mu'\rfloor)$, so that $\operatorname{dist}(A_{j_k},B_k^c)=\lfloor \log({n})/\mu'\rfloor$, and $r=\lfloor \log(n)/\nu\rfloor$, we have shown that, given $1/\nu':=1/\mu'+1/\nu$,
		\begin{align*}
			|\partial_{x'_k}f_L(\beta,x)|\le 2\beta\,h\,\|L\|_{\operatorname{Lip}}  \Big( c'(2r_0)^D\,h\,+(2(r_0+\log(n)/\nu'))^D(1+C)\Big)
		\end{align*}
		The result follows.
	\end{proof}

	With the bound of \Cref{prop:bounds}, we show that for Gibbs states belonging to a phase with exponentially decaying correlations, the difference of expected values of Lipschitz observables in two such states is controlled by the $\ell_1$-norm of their associated parameters. 
	
	% we can address two problems. The first application consists in a bound on the $\ell_2$-norm of the gradient of the function $f_L$. 
	% \begin{proposition}
		% Assume that $\{\sigma(\beta,x)\}_{x\in\Omega}$ satisfies an exponential decay of correlations. \JDW{Does this mean $g(.)$ in the previous lemma decays exponentially?} \textcolor{blue}{yes precisely} Then for any $n$ qudit observable $L$:
		% \begin{align*}
			%       \|\nabla_xf_L(x)\|_2^2 \le c'\,n^2\,\|L\|_{\operatorname{Lip}}^2\,\log(n)^2\,,
			% \end{align*}
		% for some constant $c'$ independent of $n$. 
		% \end{proposition}
	
	% \begin{proof}
		% Follows directly from Equation (\Cref{eq:partialu}).
		% \end{proof}
	\begin{corollary}\label{cor:continuity_W1}
		With the conditions of \Cref{prop:bounds}, for any $x,y\in[-1,1]^m$,
		\begin{align}\label{upperboundW11}
			W_1(\sigma(\beta,x),\sigma(\beta,y))\le  \,\|x-y\|_{\ell_1}\,\operatorname{polylog}(n)\,.
		\end{align}
		Furthermore, this inequality is tight up to a $\operatorname{polylog}(n)$ factor for $\beta=\Theta(1)$.
	\end{corollary}
	\begin{proof}
		To get the upper bound \Cref{upperboundW11}, it suffices to interpolate between the two states as follows: for any Lipschitz observable $L$, and a path $x(s)=(1-s)x+sy$, 
		\begin{align*}
			|    \tr\left[L(\sigma(\beta,x)-\sigma(\beta,y))\right] |\le\,\sum_{k=1}^m |x_k'-y_k'|\,\int_0^1\,|\partial_{k}f_L(\beta,x)|\, ds\,.
		\end{align*}
		The result follows from using \Cref{eqpartiali} above, and using the resulting inequality in the definition of Wasserstein distance, \cref{Def:Wasserstein_Distance}. 
		
		To see that the inequality is tight up to the $\operatorname{polylog}(n)$ factor, consider the family of Hamiltonians $H(x)=\sum_ix_iZ_i$, which gives rise to diagonal, product Gibbs states that clearly satisfy exponential decay of correlations. We then have:
		\begin{align}
			W_1(\sigma(\beta,x),\sigma(\beta,y))\geq \frac{1}{2}\tr\left[\sum_iZ_i(\sigma(\beta,x)-\sigma(\beta,y))\right],
		\end{align}
		as $\sum_iZ_i$ has Lipschitz constant $2$. A simple computation shows that:
		\begin{align}
			\frac{1}{2}\tr\left[\sum_iZ_i(\sigma(\beta,x)-\sigma(\beta,y))\right]=\frac{1}{2}\sum_i\left(\frac{e^{-\beta x_i}}{e^{-\beta x_i}+e^{+\beta x_i}}-\frac{e^{-\beta y_i}}{e^{-\beta y_i}+e^{+\beta y_i}}\right).
		\end{align}
		We will assume without loss of generality that $x_i<y_i$ (as otherwise we can consider the observable with $-Z_i$ instead). Under this condition, the summands are all positive and thus:
		\begin{align}
			\frac{1}{2}\tr\left[\sum_iZ_i\sigma(\beta,x)-\sigma(\beta,y))\right]=\frac{1}{2}\sum_i\left|\frac{e^{-\beta x_i}}{e^{-\beta x_i}+e^{+\beta x_i}}-\frac{e^{-\beta y_i}}{e^{-\beta y_i}+e^{+\beta y_i}}\right|.
		\end{align}
		Yet another simple computation shows that the derivative of the function $y\mapsto \frac{e^{-\beta y}}{e^{-\beta y}+e^{+\beta y}}$ is given by 
		\begin{align}\label{equ:derivative_func}
			-\frac{\beta}{2}\operatorname{sech}(\beta y).
		\end{align}
		Let $c_\beta$ denote the minimum of the function in \Cref{equ:derivative_func} for a fixed $\beta=\Theta(1)$ over $y\in[-1,1]$. Then, by the mean value theorem:
		\begin{align}
			\frac{1}{2}\sum_i\left|\frac{e^{-\beta x_i}}{e^{-\beta x_i}+e^{+\beta x_i}}-\frac{e^{-\beta y_i}}{e^{-\beta y_i}+e^{+\beta y_i}}\right|\geq \frac{c_\beta}{2}\sum_i\left|x_i-y_i\right|,
		\end{align}
		from which we conclude that:
		\begin{align}
			W_1(\sigma(\beta,x),\sigma(\beta,y))\geq\frac{c_\beta}{2}\|x-y\|_{\ell_1}.
		\end{align}
		
	\end{proof}

	We next prove that when given a local observable $O$ supported on a ball $S\subset \Lambda$ of diameter at most $k_0$ around site $i$ of the lattice, to study its behaviour as $H(x)$ varies for Gibbs states, it is sufficient to only consider the components of $x$ which parameterise local terms which are geometrically close to the observable $O$ (up to some small error).  
	
	Before we prove this, we remember that we denote by $x|_{\mathcal{S}(r)}$ the concatenation of vectors $x_j$ corresponding to interactions $h_j$ supported on regions intersecting $S(r):=\{i\in\Lambda|\,\operatorname{dist}(i,S)\le r\}$.

	\begin{lemma}[Gibbs local indistinguishability] \label{Lemma:Local_Perturbation}
		Assuming the exponential decay of correlations in \Cref{decaycorr}, then for any observable $O$ supported on region $S$, any $r\in\mathbb{N}$,  denoting $f_{O}(x):=\tr[O\,\sigma(\beta,x)]$ and identify $x|_{\mathcal{S}(r)}$ with the vector $(x|_{\mathcal{S}(r)},0_{\mathcal{S}(r)^c})\in [-1,1]^{m}$, then the following bound holds:
		\begin{align*}
			\sup_{x\in[-1,1]^m} |f_{O}(x)-f_{O}(x|_{\mathcal{S}(r)}) | \le C_1\,e^{-\frac{r}{2\xi}}\,\|O\|_\infty\,,
		\end{align*}
		for $\mathcal{O}(1)$ constants $C_1,\xi>0$ independent of $n$. In other words:
		\begin{align}\label{eq:tracelocaldistance}
			\sup_{x\in[-1,1]^m}  \|\tr_{S^c}(\sigma(\beta,x)-\sigma(\beta,x|_{\mathcal{S}(r)}))\|_1\le C_1\,e^{-\frac{r}{2\xi}}\,.
		\end{align}
		
	\end{lemma}
	
	\begin{proof}
		We identify $x|_{\mathcal{S}(r)}$ with the vector $(x|_{\mathcal{S}(r)},0_{\mathcal{S}(r)^c})\in [-1,1]^{m}$. Given the path $x(s)=(1-s)x+sx|_{\mathcal{S}(r)}$ with components $\{x_l'(s)\}_{l=1}^m$, we get
		
		\begin{align}
			\label{eqlocalization}   | f_{O}(x)-f_{O}(x|_{\mathcal{S}(r)}) |&\le \sum_{l\in \mathcal{S}(r)^c}|x_l'(0)|\,\int_0^1\,\big|\partial_{l} \tr\big[ O\sigma(\beta,x(s))\big]\big|\,ds\\
			&=\beta  \sum_{l\in \mathcal{S}(r)^c}|x'_l(0)|\,\int_0^1\,\big|\operatorname{Cov}_{\sigma(\beta,x(s))}\big(O,\widetilde{H}_l(x(s))\big)|\,ds\,,\nonumber
		\end{align}
		
		for $\widetilde{H}_{l}(x):=\Phi_{H(x)}(\partial_{l}H(x))$, where the second line comes from \Cref{derivativefL}. 
		Next, we call $j_l\in\Lambda$ the unique site such that $x_l'$ is a coordinate of $x_{j_l}$, and denote $A_{j_l}$ be the support of $h_{j_l}$. Now, the above covariance is small if $r$ is large enough, since $\widetilde{H}_j(x(s))$ can be well approximated by an observable on $S^c$. 
		Indeed,
		\begin{align*}
			\partial_l H(x)=\partial_lh_{j_\ell}\,,
		\end{align*}
		where $j_l$ denotes the index of interaction $h_{j_l}$ which depends on variable $x'_l$. 
		Therefore, whenever ${A}_{j_l}\cap {S}=\emptyset$, we proceed similarly to \Cref{prop:bounds}: given a region $B_l\supset A_{j_l}$ such that $B_{l}\cap S=\emptyset$, denoting the observable $$O_{B_l}:=\Phi_{H_{B_l}(x)}(\partial_{x'_l}h_{j_l})-\langle\Phi_{H_{B_l}(x)}(\partial_{x'_l}h_{j_l})\rangle_{\beta, x}\,,$$ we have by \Cref{LemmaapproxPhi} as well as the  assumption that the state $\sigma(\beta,x)$ has exponential decay of correlations we have the following (refer to \Cref{Fig:Exp_Decay_Bounds} for a diagram of the regions):
		\begin{align*}
			\operatorname{Cov}_{\sigma(\beta,x)}(O,\,\widetilde{H}_{l}(x)) &= \operatorname{Cov}_{\sigma(\beta,x)}(O,\,\widetilde{H}_{l}(x)-O_{B_l})+\operatorname{Cov}_{\sigma(\beta,x)}(O,\,O_{B_l})\\ 
			&\le 2 \|O\|_{\infty}\,\|\Phi_{H(x)}(\partial_{x'_l}h_{j_l})-\Phi_{H_{B_l}(x)}(\partial_{x'_l}h_{j_l})\|_\infty+ \operatorname{Cov}_{\sigma(\beta,x)}(O,\,O_{B_l})\\
			&\le 2\|O\|_\infty c'|A_{j_l}|\,h\,e^{-\mu'\operatorname{dist}(A_{j_l},B_l^c)}+ 2C|S|\,\|O\|_\infty\,h\,e^{-\nu \operatorname{dist}(S,B_l)}\\
			&\le 2(C+c')\,\|O\|_\infty\,(2r_0+k_0)^D\, h\,\Big(e^{-\mu'\operatorname{dist}(A_{j_l},B_l^c)}+e^{-\nu \operatorname{dist}(S,B_l)} \Big)
		\end{align*}
		
		\begin{figure}[h!]
			\centering
			\includegraphics[width=0.8\textwidth]{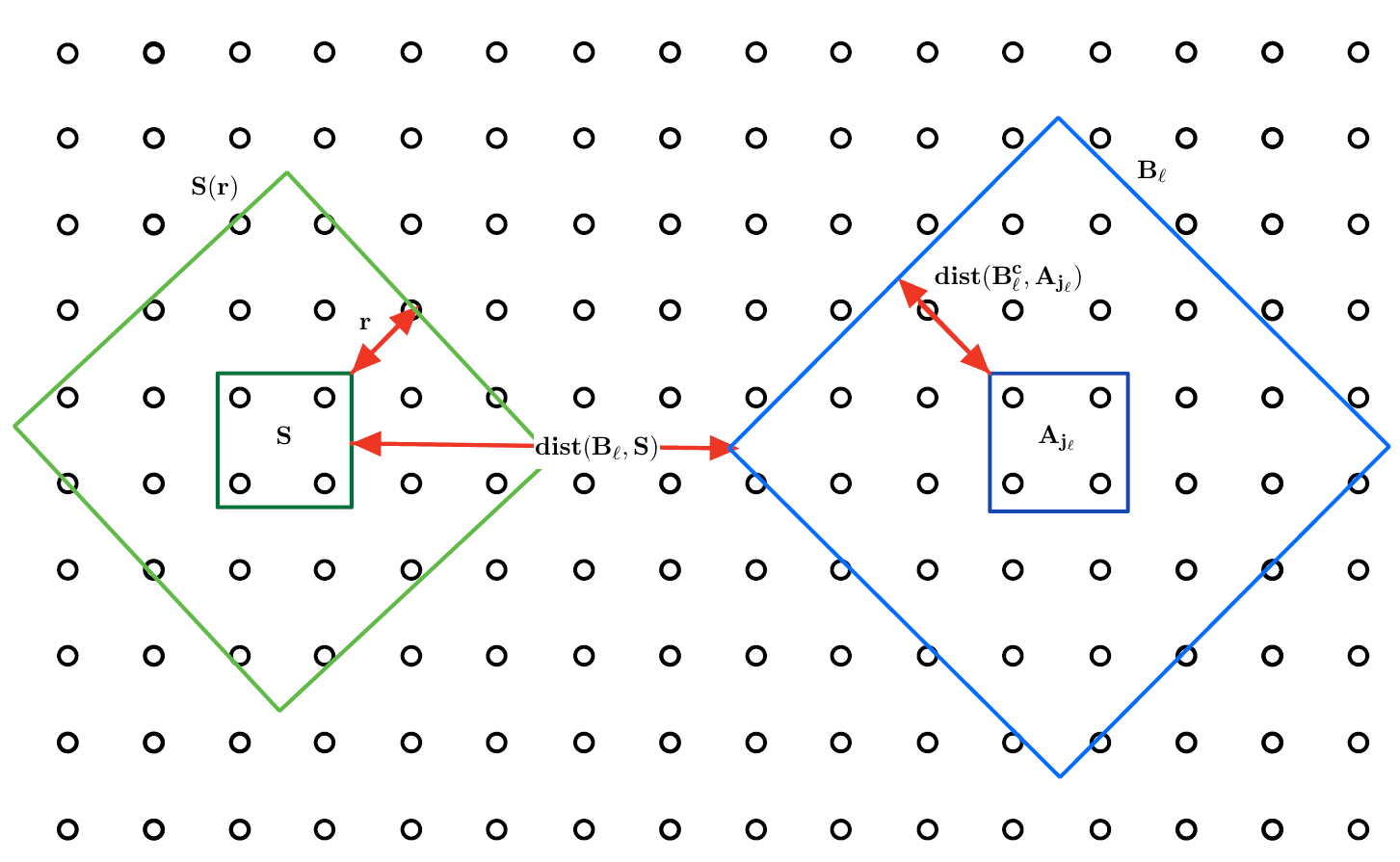}
			\caption{Diagram showing the regions involved in the proof \Cref{Lemma:Local_Perturbation}.}
			\label{Fig:Exp_Decay_Bounds}
		\end{figure}

		By construction, for $r>2r_0$, the condition that $A_{j_l}\cap S=\emptyset$ is met, and therefore the bound holds. We recall that $i\in\Lambda$ is defined as the center of $S$. Since $\operatorname{dist}(i,j_l)=k_0/2+\operatorname{dist}(S,B_l)+\operatorname{dist}(A_{j_l},B_l)+r_0$, we can choose $B_l$ so that $\operatorname{dist}(S,B_l),\operatorname{dist}(A_{j_l},B_l)\ge \operatorname{dist}(i,j_l)/2-k_0/4-r_0/2-1$. Therefore,
		\begin{align*}
			\operatorname{Cov}_{\sigma(\beta,x)}(O,\,\widetilde{H}_{l}(x))\le 4(C+c')C''\|O\|_\infty(2r_0+k_0)^D\,h e^{-\operatorname{dist}(i,j_l)/\xi}
		\end{align*}
		where $1/\xi=\min\{\mu',\nu\}$ and $C'':=e^{\max\{\mu',\nu\}(k_0/4+r_0/2+1)}$. Therefore
		\begin{align*}
			|f_{O}(x)-f_{O}(x|_{\mathcal{S}_i(r)})| 
			&\le 4\beta (C+c')\,h\,(2r_0+k_0)^D\, \|O\|_\infty\sum_{l\in \mathcal{S}(r)^c}\,e^{-\operatorname{dist}(i,j_l)/\xi}\,.
		\end{align*}
		Upon shifting the center of the lattice at site $i$, we get
		\begin{align*}
			|f_{O}(x)-f_{O}(x|_{\mathcal{S}(r)})|&\le 
			4\beta (C+c')C''\,h\,(2r_0+k_0)^D\, \|O\|_\infty\sum_{|l|\ge r+k_0/2}\,e^{-|l|/\xi}\\
			&=  4\beta (C+c')C''\,h\,(2r_0+k_0)^D\, \|O\|_\infty\sum_{a>r+k_0/2}\, \binom{a+D-1}{D-1} \,e^{-a/\xi}\\\
			&\le  4\beta (C+c')C''\,h\,(2r_0+k_0)^D\,D^{D-1} \|O\|_\infty\sum_{a>r+k_0/2}\,  a^{D-1}\,e^{-a/\xi}\\
			&\le  4\beta (C+c')C''\,h\,(2r_0+k_0)^D (D-1)!(2\xi)^{D-1}\,D^{D-1}\|O\|_\infty\sum_{a>r+k_0/2}\,  \,e^{-\frac{a}{2\xi}}\\
			&\le  4\beta (C+c')C''\,h\,(2r_0+k_0)^D(D-1)!(2\xi)^{D-1}\,D^{D-1}\|O\|_\infty  \,\frac{e^{-\frac{r+k_0/2+1}{2\xi}}}{1-e^{-\frac{1}{2\xi}}}\\
			&\equiv C_1\,e^{-\frac{r}{2\xi}}\,\|O\|_\infty\,,
		\end{align*}

		where $C_1$ depends upon all the parameters of the problem. 
		
	\end{proof}

	%The above holds for a phase centred $x_i=0$, and with the promise everywhere $x\in [-1,1]^m$ is in the phase.
	%However, we note that any bounded phase will obey a similar condition.
	%We this in mind, we get

	In the case when we are interested in distinguishing two Gibbs states with Lipschitz observables, over extended subregions of the lattice, the following extension of \Cref{cor:continuity_W1} can be easily shown to hold:

	\begin{corollary}\label{localW1bound}
		Assume that the states $\sigma(\beta,x)$ satisfy the condition of decay of correlations \Cref{decaycorr}. Then for any region $S$ of the lattice and any two $x,y\in[-1,1]^m$ 
		\begin{align*}
			W_1(\tr_{S^c}(\sigma(\beta,x)),\tr_{S^c}(\sigma(\beta,y)))\le \|x|_{\mathcal{S}(r)}-y|_{\mathcal{S}(r)}\|_{\ell_1}\,\operatorname{polylog}(|S(r)|)\,,
		\end{align*}
		where $r=\max\Big\{r_0,\, 2\xi\log\Big(\frac{2|S|C_1}{\|x|_{\mathcal{S}(r_0)}-y|_{\mathcal{S}(r_0)}\|_{\ell_1}}\Big)\Big\}$ with $r_0$ being the smallest integer such that $x|_{\mathcal{S}(r_0)}\ne y|_{\mathcal{S}(r_0)}$, and $C_1,\xi$ are the same constants as in \Cref{Lemma:Local_Perturbation}.
	\end{corollary}
	
	\begin{proof}
		Given $L_S$ a Lipchitz observable supported on region $S$ of the lattice, we have for any $r\in\mathbb{N}$:
		\begin{align*}
			\big| f_{L_S}(x)-f_{L_S}(y)\big|&\le \|L_S\|_\infty\,\Big(\big\|\tr_{S^c}\big(\sigma(\beta,x)-\sigma(\beta,x|_{\mathcal{S}(r)})\big)\big\|_1+\big\|\tr_{S^c}\big(\sigma(\beta,y)-\sigma(\beta,y|_{\mathcal{S}(r)})\big)\big\|_1\Big)\\
			&\qquad + W_1\big(\sigma(\beta,x|_{\mathcal{S}(r)}),\sigma(\beta,y|_{\mathcal{S}(r)})\big)\\
			&\le 2\, |S|\,\|L_S\|_{\operatorname{Lip}}\, C_1\,e^{-\frac{r}{2\xi}} +W_1\big(\sigma(\beta,x|_{\mathcal{S}(r)}),\sigma(\beta,y|_{\mathcal{S}(r)})\big)\,,
		\end{align*}
		where the second line follows from \Cref{eq:tracelocaldistance}. By \Cref{cor:continuity_W1}, we conclude that 
		\begin{align*}
			W_1\big(\tr_{S^c}(\sigma(\beta,x)),\tr_{S^c}(\sigma(\beta,y))\big)&\le 2\,|S|\,C_1\,e^{-\frac{r}{2\xi}} + W_1\big(\sigma(\beta,x|_{\mathcal{S}(r)}),\sigma(\beta,y|_{\mathcal{S}(r)})\big)\\
			&\le 2\,|S|\,C_1\,e^{-\frac{r}{2\xi}} + \|x|_{\mathcal{S}(r)}-y|_{\mathcal{S}(r)}\|_{\ell_1}\,\operatorname{polylog}(|S(r)|)\,.
		\end{align*}
		Next, we choose $r=2\xi\log\Big(\frac{2|S|C_1}{\|x|_{\mathcal{S}(r_0)}-y|_{\mathcal{S}(r_0)}\|_{\ell_1}}\Big)$, where $r_0$ is the smallest integer such that $x|_{\mathcal{S}(r_0)}\ne y|_{\mathcal{S}(r_0)}$. 
	\end{proof}
	
	\subsection{Gibbs state tomography beyond the trivial phase}\label{sec:beyond_trivial}
	All of the results discussed so far in this section concerned parametrized families of Gibbs states that included the maximally mixed state. Let us now discuss to what extent our results generalize to the setting of where we add a fixed Hamiltonian term $H_0$, as discussed around Eq.~\eqref{equ:shifted_Hamiltonian}.
	As made explicit in the proof of \Cref{cor:continuity_W1}, the main technical result required for the continuity of $W_1$ w.r.t. the parameters is a bound on the derivative $|  \partial_{x'_k}f_L(\beta,x)|$ in terms of the Lipschitz constant we obtained in \Cref{prop:bounds}. In turn, to prove those bounds, the main technical assumptions we required were exponential decay of correlations and the existence of LR bounds. We then have:
	\begin{corollary}\label{cor:continuity_trivial}
		Assume that $\sigma(\beta,x,H_0)$ satisfies the condition of decay of correlations, \Cref{decaycorr} and that the time evolution defined by the Hamiltonian $H(x)+H_0$ satisfies a LR-bound as in \Cref{LRbound}.
		Define the function $f_{L,H_0}(\beta,x)=\tr[L\sigma(x,\beta,H_0)]$. Then for any $k\in[m]$,
		\begin{align}\label{eqpartialiv2}
			|  \partial_{x'_k}f_{L,H_0}(\beta,x)|\le \,\|L\|_{\operatorname{Lip}}\,\operatorname{polylog}(n)\,,
		\end{align}
		for some polynomial of $\log(n)$ of degree $D$ with coefficients depending on $\beta,r_0,D,h,c',\nu,\mu'$ and $C$. In particular, it follows that 
		\begin{align}\label{equ:W_1cont2}
			W_1(\sigma(\beta,x,H_0),\sigma(\beta,y,H_0))\le  \,\|x-y\|_{\ell_1}\,\operatorname{polylog}(n)\,.
		\end{align}
	\end{corollary}
	\begin{proof}
		How \Cref{equ:W_1cont2} follows from \Cref{eqpartialiv2} is completely equivalent to \Cref{cor:continuity_W1}. 
		To obtain the inequality in \Cref{eqpartialiv2} we use the fact that
		\begin{align*}
			&\partial_{x'_k}f_{L,H_0}(\beta,x)=-\frac{\beta}{2}\,\tr\Big[L\big\{\Phi_{H(x)+H_0}(\partial_{x'_k} (H(x))),\sigma(\beta,x,H_0)\big\}\Big]+\\ &\beta \tr(\partial_{x'_k}(H(x))\sigma(\beta,x,H_0))\,\tr(L \sigma(
			\beta,x,H_0))\,.
		\end{align*}
		A close inspection of the proof of \Cref{prop:bounds} shows that all that was required to then bound the derivative was a LR-bound and the exponential decay of correlations, which we assume to also be given in this setting.
	\end{proof}
	
	The other results of the previous section also generalize to our setting with constant shift $H_0$, but we will not state them explicitly for the sake of conciseness.

	\subsection{Hamiltonian estimation and optimal Gibbs state tomography}\label{sec:learning_gibbs_examples}
	
	From \Cref{cor:continuity_W1} it is immediate that we reduced the problem of obtaining a good estimate in $W_1$ to the problem of estimating the parameters of the Gibbs state $\sigma(\beta,x)$. Indeed, it is clear that if we can obtain an estimate $x'$ of $x$ satisfying 
	\begin{align}\label{equ:required_guarantee}
		\|x-x'\|_{\ell_1}=\mathcal{O}(\epsilon n /\textrm{polylog}(n)),
	\end{align}
	then it suffices to ensure that $W_1(\sigma(\beta,x),\sigma(\beta,x'))=\epsilon n$. Let us discuss some examples where we can obtain this efficiently with $\mathcal{O}(\epsilon^{-2}\textrm{polylog}(n))$ samples.

	\subsubsection{Commuting Hamiltonians}\label{seccommuting}
	In~\cite{anshuweb}, the authors give an algorithm which with 
	\begin{align}
		e^{\mathcal{O}(\beta k^D)}\mathcal{O}(\log(\delta^{-1}n)\epsilon^{-2})
	\end{align}
	copies of $\sigma(\beta,x)$ learns $x$ up to $\epsilon$ in $\ell_\infty$ distance when $\sigma(\beta,x)$ belongs to a family of commuting, $k$-local Hamiltonians on a $D$-dimensional lattice. As we assumed that the number of parameters $m=\mathcal{O}(n)$, this translates to an algorithm with sample complexity $e^{\mathcal{O}(\beta k^D)}\mathcal{O}(\epsilon^{-2}\textrm{polylog}(\delta^{-1}n))$ to learn $x$ up to $\epsilon n$ in $\ell_1$ distance. It should be noted that the time complexity of their algorithm is $\mathcal{O}(ne^{\mathcal{O}(\beta k^D)}\epsilon^{-2}\textrm{polylog}(\delta^{-1}n))$.

	Thus, any commuting model at constant temperature satisfying exponential decay of correlations can be efficiently learned with $\operatorname{polylog}(n)$ samples. Examples of classes of commuting states that satisfy exponential decay of correlations include:
	\begin{enumerate}
		\item $1$D translation-invariant Hamiltonians at any positive temperature~\cite{araki1969gibbs}.
		\item Commuting Gibbs states of Hamiltonians on regular lattices below a threshold temperature~\cite{kliesch2014locality,harrow2020classical}.
		\item Classical spin models away from criticality~\cite{dobrushin1987completely,liu2019fisher,harrow2020classical}.
		\item Ground states of gapped systems \cite{bravyi2010topological,bravyi2011short,michalakis2013stability,nachtergaele2022quasi}.
	\end{enumerate}
	
	\subsubsection{High-temperature Gibbs states}\label{sechightemperature}
	Another class of states for which the conditions of our results hold are local Gibbs states on a lattice above a threshold temperature that depends on the locality of the Hamiltonian and the dimension of the lattice. These systems are known to have exponential decay of correlations~\cite{kliesch2014locality,harrow2020classical}. Furthermore, in~\cite{haah2021optimal} the authors give an algorithm to learn $x$ up to error $\epsilon$ in $\ell_\infty$ norm from $\mathcal{O}(\epsilon^{-2}\textrm{polylog}(\delta^{-1}n))$ samples. This again translates to a $\mathcal{O}(\epsilon n)$ error in $\ell_1$ norm. Note that their algorithm also is computationally efficient.
	
	We note that in~\cite{anshu2021sample} the authors give an algorithm to learn the Hamiltonian of any Gibbs state of positive temperature through the maximum entropy method. However, their results require a polynomial number of samples to recover the parameters in $\ell_1$ distance. Thus, their results do not work for the polylog regime investigated in this work.

	\subsubsection{Gibbs state of exponentially decaying correlations and conditional mutual information}\label{sec:gibbs_learning_exponential_decay}

	In the previous section, we extracted two regimes for which there exist efficient Gibbs tomography algorithms from previous works, namely the commuting and the high-temperature regimes. As said before, depending on the Hamiltonian, exponential decay of correlations can also occur in the low-temperature regime, and it is an interesting open question whether our strategy can be adapted to that setting for non-commuting interactions. 
	
	Here, we show that the Gibbs state $\sigma(\beta,x)$ of a possibly non-commuting Hamiltonian $H(x)$ can also be estimated in Wasserstein distance up to multiplicative error $\epsilon n$ given $\operatorname{polylog}(n)$ copies of it as long as the latter has exponentially decaying correlations and is close to a quantum Markov chain, hence partially answering an open problem previously raised in \cite{anshu2021sample}.
	
	To be more precise, in this section we will require a stronger notion of decay of correlations.
	\begin{definition}[Uniform clustering]
		The Gibbs state $\sigma(\beta,x)$ is said to be uniformly $\zeta(\ell)$-clustering if for
		any $X\subset\Lambda$ and any $A\subset X$ and $B\subset X$ 
		such that $\operatorname{dist}(A,B) \ge \ell$,
		\begin{align*}
			\operatorname{Cov}_{\sigma(\beta,x,X)}(X_A, X_B) \le  \|X_A\|_{\infty}\,\|X_B\|_{\infty}\,\zeta(\ell) 
		\end{align*}
		for any $X_A$ supported on $A$ and $X_B$ supported on $B$.
	\end{definition}
	
	As pointed out in \cite{Brandao_Kastoryano_2018}, this property is called uniform clustering to contrast with regular clustering property that usually only refers to properties of the state $\sigma(\beta,x)$.
	
	\begin{definition}[Uniform Markov condition]
		
		The Gibbs state $\sigma_\beta(x)$ is said to satisfy the uniform $\delta(\ell)$-Markov condition if for any $ABC = X \subset \Lambda$ with $B$ shielding $A$ away from $C$ and such that $\operatorname{dist}(i,j) \ge \ell$ for any $i \in A$ and $j \in C $, we have
		\begin{align*}
			I(A:C|B)_{\sigma(\beta,x,X)} \le  \delta(\ell)\,. 
		\end{align*}
	\end{definition}
	
	This property always holds for commuting Gibbs states for a function $\delta(\ell)=0$ as soon as $\ell$ is larger than twice the interaction range. Although not proven yet, it is believed that the approximate Markov property holds with some generality for non-commuting Gibbs states. The 1D and high-temperature settings were investigated in \cite{kato2019quantum} and \cite{kuwahara_gibbs}, respectively. The decay of the conditional mutual information was also shown for finite temperature Gibbs states of free fermions, free bosons, conformal field theories, and holographic models \cite{PhysRevB.94.155125}, as well as more recently for purely generated finitely correlated states in \cite{doi:10.1063/5.0085358}.

	We will now show how to learn states that satisfy both the uniform Markov condition and the uniform clustering of correlations. Our strategy consists in using the maximum entropy estimation \cite{jaynes1957information,jaynes1957informationb,jaynes1982rationale,brandao2017exponential}, already appearing in \cite{anshu2021sample}, to construct an estimator $\hat{x}$ of the parameter $x\in[-1,1]^m$. The condition of exponential decay of correlations and that of approximate Markov chain will ensure that $W_1(\sigma(\beta,\hat{x}),\sigma(\beta,x))=o(n)$. Thus, we once again emphasise that our goal is to obtain a good recovery of the state, not of the parameter $x$.

	For sake of clarity and simplicity of presentation, we only consider the $1$D setting, although our method easily extends to arbitrary dimension. We assume that each interaction $h_j(x_j)$ is of the form
	\begin{align*}
		h_j(x_j):=\sum_{l=1}^{\ell} x_{j,l} \,h_{j,l}
	\end{align*}
	%\todo[inline]{Don't we need that $\tr[h_{j,l}h_{j',l'}]=\delta_{j,j'}\delta_{l,l'}$ or something to this effect?}
	for some self-adjoint operators $h_{j,l}$ supported in $A_j:=\{k\in\Lambda|\operatorname{dist}(k,j)\le r_0\}$ with $\|h_{j,l}\|\le h$, where we denoted by $x_{j,l}$ the entries of $x_k$. We also recall that given a region $R$ of the lattice, we denote $H_R(x):=\sum_{k\in R} h_k(x_k)$. In what follows, with a slight abuse of notations, we denote by the same symbol a vector $y=\{y_{k,l}\}_{k\in N_j}$ and its embedding $(y,0)$ onto $[-1,1]^m$. Then, given an inverse temperature $\beta>0$, we define the partition function as
	\begin{align*}
		Z_{\beta}(x)=\tr\big[e^{-\beta H(x)}\big]\,.
	\end{align*}
	% where $x_{R}$ stands for the sub-vector of coordinates $x_{k,l}$, $k\in R$, $l=1,...,\ell$. We also denote $Z_{\beta}(x)=Z_{\beta,\Lambda}(x)$ for $x\in[-1,1]^m$. 
	% We pave the chain $\Lambda=[-L,L]$ with overlapping unions of intervals $S_i,T_i$ such that $|S_i\cap T_i|=|T_i\cap S_{i+1}|=r$ for some even integer $r$ to be determined later, and such that two consecutive intervals $S_i,S_{i+1}$, resp.~$T_i,T_{i+1}$, are spaced by more than twice the interaction range $r_0$. Therefore, the Gibbs states $\sigma_S(\beta,x)\propto{e^{-\beta H_S(x)}}$, resp.~$\sigma_T(\beta,x)\propto{e^{-\beta H_T(x)}}$, are effectively commuting, up to relabeling the interactions as $H_{S_i}$, resp.~$H_{T_i}$. 
	The maximum entropy problem consists in the following strongly convex optimisation problem.
	\begin{theorem}[\cite{anshu2021sample}]
		Given an unknown Hamiltonian $H(x)=\sum h_j(x_j)$, define $e_{k,l}= \tr[h_{k,\ell}\sigma(\beta, x)  ]$.
		Solving the following optimisation problem:
		\begin{align}\label{localmaxent}
			\hat{x}:=\underset{y\in[-1,1]^m}{\operatorname{arg~min}}\,L(y)\,,\qquad \text{ where }\qquad L(y):=\log Z_{\beta}(y)+\beta\sum_{k\in\Lambda}\sum_{l=1}^\ell y_{k,l}\,e_{k,l}\,
		\end{align}
		gives $\hat{x}$ such that $\sigma(\beta, \hat{x}) = \sigma(\beta, x)$.
	\end{theorem}
	
	In an experimental setting, we will not have access to the exact $\{e_{k,l}\}_{k,l}$, but instead may be able to approximate them using by having access to the state.
	However, we want to be sure that having a reasonably good approximation to $e_{k,l}$ is sufficient to approximate $x$.
	To do so one can make use of the fact that
	\begin{align}\label{upperboundZ}
		\log Z_\beta(\hat{x})\le \log Z_\beta(x)+\beta\sum_{k\in\Lambda}\sum_{l=1}^\ell\,(x_{k,l}-\hat{x}_{k,l})\,\widetilde{e}_{k,l}\,.
	\end{align}
	Further assuming $\alpha_2$ is a lower bound on the strong convexity constant associated to the function $x\mapsto \log Z_\beta(x)$, that is $\nabla^2 Z_{\beta}\ge \alpha_2 \,I$, we have by Taylor expansion and since $\partial_{x_{k,l}}\log Z_\beta(x)=-\beta e_{k,l}(x)$:
	\begin{align}\label{strongconvecl2}
		\log Z_\beta(\hat{x})\ge \log Z_\beta(x)-\beta\sum_{k\in\Lambda}\sum_{l=1}^\ell\, (\hat{x}_{k,l}-x_{k,l})\,e_{k,l}(x)+\frac{\alpha_2}{2}\,\|x-\hat{x}\|_{\ell_2}^2\,.
	\end{align}
	Combining the two bounds above, we find that
	\begin{align*}
		\|x-\hat{x}\|_{\ell_2}^2\le \frac{2\beta}{\alpha_2}\,\sum_{k,l} (x_{k,l}-\hat{x}_{k,l}) (\widetilde{e}_{k,l}-e_{k,l}(x))\le \frac{2\beta}{\alpha_2}\,\|x-\hat{x}\|_{\ell_2}\,\|e-\widetilde{e}\|_{\ell_2}\,,
	\end{align*}
	and hence $\|x-\hat{x}\|_{\ell_2}\le \frac{2\beta\sqrt{\ell |\Lambda|}\,\eta}{\alpha_2}$, thus giving the following theorem:
	\begin{theorem}[\cite{anshu2021sample}] \label{Theorem:l2_Norm_Estimate}
		Suppose $\widetilde{e}_{k,l}$ is an approximation of $e_{k,l}(x):=\tr \big[h_{k,l}\,\sigma(\beta,x)\big]$ with $\|\widetilde{e}-e(x)\|_{\ell_\infty}\le \eta$.
		Assume that the following inequality is satisfied for some $\alpha_2$: $\nabla^2 Z_\beta \geq \alpha I$.
		Solving the following optimisation problem:
		\begin{align}\label{localmaxent}
			\hat{x}:=\underset{y\in[-1,1]^m}{\operatorname{arg~min}}\,L(y)\,,\qquad \text{ where }\qquad L(y):=\log Z_{\beta}(y)+\beta\sum_{k\in\Lambda}\sum_{l=1}^\ell y_{k,l}\,\widetilde{e}_{k,l}\,
		\end{align}
		gives an output $\hat{x}$ satisfying:
		\begin{align*}
			\norm{\hat{x}-x}_{\ell_2}\leq \frac{2\beta\eta \sqrt{\ell \Lambda}}{\alpha_2}.
		\end{align*}
	\end{theorem}
	% where $\mathring{R}=\cup_{i}\mathring{R}_i$ with $\mathring{R}_i$ defined as the interior of $R_i$: $\mathring{R}_i:=\{j\in R_i|\,\operatorname{dist}(j,R_i^c)\ge \frac{r}{2}\}$, so that $\Lambda=\mathring{S}\cup\mathring{T}$. 
	%In the optimisation \eqref{localmaxent},  In \cite{anshu2021sample}, this optimisation was used to find a bound on the $\ell_2$-distance between $x$ and $\hat{x}$ as follows. First, by definition of $\hat{x}$, $L(\hat{x})\le L(x)$, or equivalently
	%\begin{align}\label{upperboundZ}
	%\log Z_\beta(\hat{x})\le \log Z_\beta(x)+\beta\sum_{k\in\Lambda}\sum_{l=1}^\ell\,(x_{k,l}-\hat{x}_{k,l})\,\widetilde{e}_{k,l}\,.
	%\end{align}
	
	%On the other hand, denoting by $\alpha_2$ a lower bound on the strong convexity constant associated to the function $x\mapsto \log Z_\beta(x)$, that is $\nabla^2 Z_{\beta}\ge \alpha_2 \,I$, we have by Taylor expansion and since $\partial_{x_{k,l}}\log Z_\beta(x)=-\beta e_{k,l}(x)$:
	%\begin{align}\label{strongconvecl2}
	%    \log Z_\beta(\hat{x})\ge \log Z_\beta(x)-\beta\sum_{k\in\Lambda}\sum_{l=1}^\ell\, (\hat{x}_{k,l}-x_{k,l})\,e_{k,l}(x)+\frac{\alpha_2}{2}\,\|x-\hat{x}\|_{\ell_2}^2\,.
	%\end{align}
	%Combining the two bounds above, we find that
	%\begin{align*}
	%\|x-\hat{x}\|_{\ell_2}^2\le \frac{2\beta}{\alpha_2}\,\sum_{k,l} (x_{k,l}-\hat{x}_{k,l}) (\widetilde{e}_{k,l}-e_{k,l}(x))\le \frac{2\beta}{\alpha_2}\,\|x-\hat{x}\|_{\ell_2}\,\|e-\widetilde{e}\|_{\ell_2}\,,
	%\end{align*}
	%and hence $\|x-\hat{x}\|_{\ell_2}\le \frac{2\beta\sqrt{\ell |\Lambda|}\,\eta}{\alpha_2}$. 
	Using the bound on $\hat{x}$ from \Cref{Theorem:l2_Norm_Estimate} the equivalence between $\ell_1$ and $\ell_2$-norms, we have that
	\begin{align*}
		\|x-\hat{x}\|_{\ell_1}\le \frac{2\beta\ell n\eta}{\alpha_2}\,,
	\end{align*}
	which provides us with the right scaling for our $\ell_1$ approximation problem as long as $\eta=o(1)$ and $\alpha_2=\Omega(1)$. 
	Unfortunately, the constant $\alpha_2$ could only be proved to scale inverse polynomially with $n$ in \cite{anshu2021sample}. 
	A first idea from there is to try and find a constant $\alpha_1=\Omega(n^{-1})$ such that the following strong convexity bound with respect to the $\ell_1$-norm holds.
	As per \cref{strongconvecl2}, this would imply:
	\begin{align}\label{strongconvexl1}
		\log Z_\beta(\hat{x})\ge \log Z_\beta(x)-\beta\sum_{k,l}(\hat{x}_{k,l}-x_{k,l})\,e_{k,l}(x)+\frac{\alpha_1}{2}\,\|x-\hat{x}\|_{\ell_1}^2\,.
	\end{align}
	If such a bound held, we would conclude similarly to the previous setting that
	\begin{align*}
		\|x-\hat{x}\|_{\ell_1}\le \frac{2\beta \eta}{\alpha_1}=o(\eta n)\,.
	\end{align*}
	Which together with the continuity bound \Cref{upperboundW11} would allow us to get the desired recovery estimate in Wasserstein distance. Now, it can be seen that \Cref{strongconvexl1} is equivalent to 
	\begin{align}\label{eq_TCparameters}
		\|x-\hat{x}\|_{\ell_1}^2\le\,\frac{2}{\alpha_1}\, D(\sigma(\beta,x)\|\sigma(\beta,\hat{x}))\,.
	\end{align}
	Here we recall that the relative entropy between two quantum states $\rho$ and $\sigma$ with $\operatorname{supp}(\rho)\subseteq \operatorname{supp}(\sigma)$ is $D(\rho\|\sigma):=\tr\rho\log\rho-\tr\rho\log\sigma$.
	This together with \Cref{upperboundW11} would lead to the following local version of the \textit{transportation cost} inequality 
	\begin{align}\label{TC}
		W_1(\sigma(\beta,x),\sigma(\beta,\hat{x}))^2\le \mathcal{O}(n\operatorname{polylog}(n))\,D(\sigma(\beta,x)\|\sigma(\beta,\hat{x}))\,.
	\end{align}
	In \cite{DePalma2022}, such inequality was shown to hold in the high-temperature regime only for commuting $H$, albeit when $\sigma(\beta,x)$ can be replaced by an arbitrary state $\rho$ on the lattice. The latter is referred to as a transportation-cost inequality for the state $\sigma(\beta,\hat{x})$. Since \Cref{strongconvexl1} consists in a strengthening of 
	\Cref{TC}, proving it directly appears difficult. Here instead, we want to show the following weakening of \eqref{TC}:
	\begin{align*}
		W_1(\sigma(\beta,x),\sigma(\beta,\hat{x}))^2\le \mathcal{O}(n\operatorname{polylog}(n))\,D(\sigma(\beta,x)\|\sigma(\beta,\hat{x}))\,+ o(\epsilon n)\,,
	\end{align*}
	for some constant $\delta$ which depends on the approximate Markov as well as the correlation decay properties of the Gibbs state $\sigma(\beta,\hat{x})$. More precisely, we show the following extension of \cite[Theorem 4]{DePalma2022} to Gibbs states of non-commuting Hamiltonians.
	\begin{proposition}[Generalised transportation-cost inequality]\label{prop:generalized_transportation}
		With the notations of the above paragraph, for all states $\rho$:
		\begin{align*}
			W_1(\rho,\sigma(\beta,x))\,\le\, \inf_{\ell\in\mathbb{N}}\mathcal{O}(\ell \sqrt{n})\,\sqrt{D(\rho\|\sigma(\beta,x))} + n^2\big(\delta(\mathcal{O}(\ell))+\zeta(\mathcal{O}(\ell))+e^{-\mathcal{O}(\ell)}\big)\,.
		\end{align*}
		In particular, if both $\zeta(l),\delta(l)=\mathcal{O}(e^{-\xi l})$, then for $l=\mathcal{O}(\xi^{-1}\log(n\epsilon^{-1}))$ we have
		\begin{align}\label{equ:final_w1_exp}
			W_1(\rho,\sigma(\beta,x))\,\le\, \mathcal{O}(\log(n\epsilon^{-1}) \sqrt{n})\,\sqrt{D(\rho\|\sigma(\beta,x))}+o(\epsilon n).
		\end{align}
	\end{proposition}
	
	\begin{proof}
		The proof is adapted from that of \cite[Theorem 4]{DePalma2022}. We first consider a bipartite quantum subsystem $AB\subset \Lambda$ and a joint quantum state $\omega_{AB}$ of $AB$. We then define the so-called quantum recovery map \cite{sutter2017multivariate,junge2018universal} by its action on a quantum state $\rho_A$ on region $A$:
		\begin{equation}\label{eq:Phirecov}
			\Phi_{A\to AB}(\rho_A) = \int_\mathbb{R}\omega_{AB}^{\frac{1-it}{2}}\,\omega_A^\frac{it-1}{2}\,\rho_A\,\omega_A^{-\frac{1+it}{2}}\,\omega_{AB}^{\frac{1+it}{2}}\,d\mu_0(t)\,,
		\end{equation}
		where $\mu_0$ is the probability distribution on $\mathbb{R}$ with density
		\begin{equation}
			d\mu_0(t) = \frac{\pi\,dt}{2\left(\cosh(\pi t)+1\right)}\,.
		\end{equation}
		If $A$ is in the state $\omega_A$, the recovery map $\Phi_{A\to AB}$ recovers the joint state $\omega_{AB}$, \emph{i.e.}, $\Phi_{A\to AB}(\omega_A) = \omega_{AB}$.
		The relevance of the recovery map comes from the recoverability theorem \cite{sutter2017multivariate}, which states that $\Phi_{A\to AB}$ can recover a generic joint state $\rho_{AB}$ from its marginal $\rho_A$ if removing the subsystem $B$ does not significantly decrease the relative entropy between $\rho$ and $\omega$.
		More precisely, for any quantum state $\sigma_{AB}$ of $AB$ we have
		\begin{equation}\label{eq:recovery}
			D(\sigma_{AB}\|\omega_{AB}) - D(\sigma_A\|\omega_A) \ge D_{\mathbb{M}}(\sigma_{AB}\|\Phi_{A\to AB}(\sigma_A))\,,
		\end{equation}
		where $D_{\mathbb{M}}$ denotes the measured relative entropy \cite{donald1986relative,petz1986sufficient,hiai1991proper,berta2017variational}
		\begin{align}
			D_{\mathbb{M}}(\sigma\|\omega):=\sup_{(\mathcal{X},M)}\,D(P_{\sigma,M}\|P_{\omega,M})\,,
		\end{align}
		where the supremum above is over all positive operator valued measures $M$ that map the input quantum state to a probability distribution on a finite set $\mathcal{X}$ with probability mass function given by $P_{\rho,M}(x)=\tr\rho M(x)$.
		
		Next, we split region $A$ into regions $A_1$ and $A_2$ such that $A_1$ shields $A_2$ away from $B$, and take $\sigma_{AB}:=\tr_{(AB)^c}(\sigma(\beta,x))$ and $\omega_{AB}=\sigma_{A_1B}\otimes\sigma_{A_2}$ In that case, \eqref{eq:recovery} becomes 
		\begin{align}
			I(B:A_2|A_1)_{\sigma}\ge D_{\mathbb{M}}\big(\sigma\|\Phi_{A_1\to A_1B}(\sigma_A)\big)\,,
		\end{align}
		where we also used that the state $\omega$ is a tensor product in the cut $A_1B-A_2$, so that $\Phi_{A\to B}=\Phi_{A_1\to B}$. 
		
		Next, we pave the chain $\Lambda$ into unions of intervals $A=\cup_{i=1}^M A_i$ and $B=\cup_{i=1}^MB_i$ such that $A_i\cap B_i\ne \emptyset$ and $B_i\cap A_{i+1}\ne \emptyset$. As in \cite{Brandao_Kastoryano_2018}, we then define the channel $\mathcal{F}:=\mathcal{F}_A\circ \mathcal{F}_B$ where $\mathcal{F}_B:=\bigotimes_{i}\sigma(\beta,x,B_i)\otimes\tr_{B_i}$ and $\mathcal{F}_A:=\prod_j \Phi_{A_i\backslash B\to A_i}\circ \tr_{A_i}$. In words, the channel $\mathcal{F}_B$ first prepares the Gibbs state in the region $B$, whereas $\mathcal{F}_A$ prepares the remaining of the Gibbs state onto region $A\backslash B$.
		Then, we have, for any state $\rho$
		\begin{align*}
			W_1(\rho,\sigma(\beta,x))&\le W_1(\rho,\mathcal{F}_B(\rho))+W_1(\mathcal{F}_B(\rho),\mathcal{F}_A\circ\mathcal{F}_B(\rho))+W_1(\mathcal{F}(\rho),\sigma(\beta,x))\\
			&\le \sum_{i}W_1(\sigma(A,i),\sigma(A,i+1))+\sum_i\,W_1(\sigma(B,i),\sigma(B,i+1))+n\,\|\mathcal{F}(\rho)-\sigma(\beta,x)\|_1\\
			&\overset{(1)}{\le} R\,\sum_i \,\|\sigma(A,i)-\sigma(A,i+1)\|_1+\|\sigma(B,i)-\sigma(B,i+1)\|_1+n\|\mathcal{F}(\rho)-\sigma(\beta,x)\|_1\,,
		\end{align*}
		where $\sigma(A,i):=\bigotimes_{j<i}\sigma(\beta,x,B_j)\otimes\tr_{B_j}(\rho)$ and $\sigma(B,i):=\bigotimes_{j<i}\Phi_{A_j\backslash B,\to A_j}\circ \tr_{A_j}(\mathcal{F}_B(\rho))$, and where $R=\max\{|B_i|,|A_i|\}$ so that $(1)$ follows from \cite[Propositiom 5]{de2021quantum}. Next, we use Pinsker's inequality, so
		\begin{align*}
			&W_1(\rho,\sigma(\beta,x))\\
			&\qquad \le R\sqrt{2}\,\sum_i \sqrt{D_{\mathbb{M}}(\sigma(A,i)\|\sigma(A,i+1))}+\sqrt{D_{\mathbb{M}}(\sigma(B,i)\|\sigma(B,i+1))}\,+n\,\|\mathcal{F}(\rho)-\sigma(\beta,x)\|_1\\
			&\qquad  \le 2R\sqrt{M}\,\left(\sum_{i}D_{\mathbb{M}}(\sigma(A,i)\|\sigma(A,i+1))+D_{\mathbb{M}}(\sigma(B,i)\|\sigma(B,i+1))\right)^{\frac{1}{2}}+n\|\mathcal{F}(\rho)-\sigma(\beta,x)\|_1\\
			&\qquad \overset{(1)}{\le} 2R\sqrt{M} \,\sqrt{D(\rho\|\sigma(\beta,x))}+n\|\mathcal{F}(\rho)-\sigma(\beta,x)\|_1\\
			&\qquad\overset{(2)}{\le} 2R\sqrt{M}\, \sqrt{D(\rho\|\sigma(\beta,x))} + n^2\big(\delta(\mathcal{O}(R))+\zeta(\mathcal{O}(R))+C_1e^{-c_2 R}\big)\,,   
		\end{align*}
		where $(1)$ follows from multiple uses of \eqref{eq:recovery} as well as the sub-additivity of the relative entropy under tensor products in its second argument, whereas $(2)$ comes from \cite[Theorem 6]{Brandao_Kastoryano_2018}. The result follows.
	\end{proof}
	
	We can then easily turn the previous statement into one about learning the state $\sigma(\beta,x)$:
	\begin{corollary}[$W_1$ learning from the uniform Markov condition]\label{coroapproxmarkov}
		Under the same conditions of \Cref{prop:generalized_transportation} assume further that $\zeta(l),\delta(l)=\mathcal{O}(e^{-\xi l})$. Then $\mathcal{O}(\epsilon^{-4}\operatorname{polylog}(n\delta^{-1}))$ samples of $\sigma(\beta,x)$ suffice to learn a state $\sigma(\beta,x')$ s.t. with probability at least $1-\delta$
		\begin{align}
			W_1(\sigma(\beta,x),\sigma(\beta,x'))=\mathcal{O}(\epsilon n).
		\end{align}
	\end{corollary}
	\begin{proof}

		We simply need to adapt the proof of \Cref{eq_TCparameters}. First, we recall that 
		\begin{align*}
			D(\sigma(\beta,x)\|\sigma(\beta,\hat{x}))=\beta\sum_{k,l}(\hat{x}_{k,l}-x_{k,l})\,e_{k,l}(x)+\log Z_\beta(\hat{x})-\log Z_\beta(x)\,.
		\end{align*}
		Together with \Cref{equ:final_w1_exp}, we have the following approximate strong convexity bound for the log partition function in the Wasserstein topology:
		\begin{align*}
			W_1(\sigma(\beta,x),\sigma(\beta,\hat{x}))^2\le \mathcal{O}(\log(n\epsilon^{-1})^2 {n})\,D(\rho\|\sigma(\beta,x))+o(\epsilon^2 n^2)
		\end{align*}
		Combining with \Cref{upperboundZ} and assuming that $\|\widetilde{e}-e(x)\|_{\ell_\infty}\le \eta$, we get 
		\begin{align*}
			W_1(\sigma(\beta,x),\sigma(\beta,\hat{x}))^2&\le \mathcal{O}(\log(n\epsilon^{-1})^2 {n})\,\beta\,\sum_{k,l}(\hat{x}_{k,l}-x_{k,l})(e_{k,l}(x)-\widetilde{e}_{k,l}) +o(\epsilon^2 n^2)\\
			&\le  \mathcal{O}(\log(n\epsilon^{-1})^2 {n^2})\,\eta\,+o(\epsilon^2 n^2)\,.
		\end{align*}
		Above, we have managed to reduce the problem to that estimating the coefficients $e_{k,l}$ to precision $\eta=\mathcal{O}(\epsilon^2/\operatorname{polylog}(n))$. This can be done with probability $1-\delta$ given $\mathcal{O}(\epsilon^{-4}\operatorname{polylog}(n\delta^{-1}))$ copies of the state $\sigma(\beta,x)$ through various methods (see e.g.~\cite[Appendix A]{rouze2021learning} or \cite[Corrolary 27]{anshu2021sample} for more details).

		% --------------------------------------------------------------------------------------------------------------
		
		% The first step of the proof is to recall the identity
		% \begin{align}\label{equ:sym_relative_entropy}
			% &D(\sigma(\beta,x)\|\sigma(\beta,\hat{x}))+D(\sigma(\beta,\hat{x})\|\sigma(\beta,x))=\beta|\langle x-\hat{x},e(x)-e(\hat{x})\rangle|\leq\\ &\beta\|x-x'\|_{\infty}\|e(x)-e(x')\|_{\ell_1},
			% \end{align}
		% where once again $e(x)_j=\tr\left[h_k\sigma(\beta,x)\right]$ and the last inequality is H\"older's. Thus, it follows from \Cref{equ:final_w1_exp} combined with \Cref{equ:sym_relative_entropy} that if we find a $x'$ s.t. 
		% \begin{align}\label{equ:good_ell1}
			% \|e(x)-e(x')\|_{\ell_1}\leq \mathcal{O}(\epsilon^2/\operatorname{polylog}(n)),
			% \end{align}
		% then
		% \begin{align}
			% W_1(\sigma(\beta,x),\sigma(\beta,x')\leq \mathcal{O}(\epsilon n).
			% \end{align}
		
	\end{proof}
	
	\begin{remark} The main difference between the approach outlined in the proposition above and that of \Cref{cor:continuity_W1} is that, besides requiring the additional assumption of uniform Markovianity, it has an $\epsilon^{-4}$ scaling with the precision instead of $\epsilon^{-2}$. However, it completely bypasses the need for a good algorithm to learn the parameters $x$, which is required to apply \Cref{cor:continuity_W1}. Indeed, note that all that was required was to find another Gibbs state that approximates the local expectation values $e_{k,l}(x)$, which can be achieved through the maximum entropy principle. However, to the best of our knowledge, it is not known how to run the maximum entropy principle efficiently for the classes of states considered in \Cref{cor:continuity_W1} on a classical computer. On the other hand, the results of~\cite{Brandao_Kastoryano_2018} show that we can run maximum entropy efficiently for the states under consideration for \Cref{cor:continuity_W1} on a quantum computer.
		
		Conversely, if combined with an $\ell_2$ strong convexity guarantee on the partition function, the $\epsilon^{-2}$ scaling can be recovered (see e.g.~\cite[Theorem A.1]{rouze2021learning} for more details). But, since we are currently only able to get this guarantee for commuting Hamiltonians or in the high-temperature regimes, \Cref{cor:continuity_W1} provides a more direct path than the strategy previously exhibited in \cite{rouze2021learning}.
	\end{remark}

	\section{Algorithm for learning observables in high temperature phases}\label{seclearningalgogibbs}
	
	Next, we assume we are given different Gibbs states $\sigma(\beta,x)$ where $x$ is sampled according to the uniform distribution $U$ over $x\in [-1,1]^m$, and wish to learn the expectation value of an unknown observable for all values of $x\in [-1,1]^m$.
	We assume that the Gibbs states in the interval $[-1,1]^m$ has exponentially decaying correlations everywhere, which can be thought of as defining a continuous phase of matter.
	For commuting Hamiltonians, this relationship can be made more precise \cite{harrow2020classical}.
	
	As before, for pedagogical reasons we will first solve this problem for phases that contain the maximally mixed state in them and, thus, we can compute local expectations efficiently on a classical computer. But then we will show how these results easily generalize to phases that are more computationally involved by once again shifting the center of the phase. We will discuss these more involved examples in~\Cref{sec:beyond_high_T}.

	During a training stage, we pick $N$ points $Y_1,\dots , Y_N\sim U$ independently distributed uniformly at random in $[-1,1]^m$ and are given access to the Gibbs states $\sigma(\beta,Y_j)$. 
	Next, fix $r\in\mathbb{N}$. 
	Given an observable $O=\sum_{i=1}^M O_i$, we define $S_i=\operatorname{supp}(O_i)$ and 
	for each $S_i$ there is a ball of diameter at most $k_0$ containing $S_i$, and $S_i(r):=\{j\in\Lambda|\operatorname{dist}(j,S_i)\le r\}$. 
	We construct for any $x\in [-1,1]^m$ the estimator 
	\begin{align}\label{equ:defin_estimator}
		\hat{f}_O(x)=\sum_{i=1}^M\tr \big[O_i\, \sigma(\beta, \hat{Y}_i(x))\big]\,,\quad \text{ with }\quad \hat{Y}_i(x)=\operatorname{argmin}_{{Y}_k}\|x|_{\mathcal{S}_i(r)}-{Y}_k|_{\mathcal{S}_i(r)}\|_{\ell^\infty}\,,
	\end{align}
	where we recall that we denote by $x_{\mathcal{S}_{i}(r)}$ the concatenation of vectors $x_j$ corresponding to interactions $h_j$ supported on regions intersecting $S_i(r)$. In words, we approximate the expectation value of $O_i$ by that of the Gibbs state whose parameters in a region around $S_i$ are the closest to the state of interest. We also denote $\mathcal{S}_i\equiv \mathcal{S}_i(0)$.

	\Cref{Lemma:Local_Perturbation} demonstrates that to measure a particular observable, it is sufficient to consider the parameters only spatially local to it.
	One important technical tool for that is the following Lemma from~\cite{brandao2017quantum}:
	\begin{lemma}[Lemma 16 of ~\cite{brandao2017quantum}]\label{lem:Gibbs_perturbation}
		Let $H,H'$ be Hermitian matrices. Then:
		\begin{align}
			\Big\|\frac{e^{-H_1}}{\tr\left[e^{-H_1}\right]}-\frac{e^{-H_2}}{\tr\left[e^{-H_2}\right]}\Big\|\leq 2(e^{\|H_1-H_2\|_\infty}-1).
		\end{align}
	\end{lemma}

	Next, we show that if we are given a number of samples uniformly sampled from $[-1,1]^m$, we can construct an estimator for the observable to high precision.

	\begin{proposition}\label{approxbetax}
		Assuming the exponential decay of correlations in \Cref{decaycorr}, the estimator $\hat{f}_O(x):=\sum_{i}\,\tr[O_i\,\sigma(\beta,\hat{Y}_i(x))]$ satisfies the bound
		\begin{align*}
			\sup_{x\in[-1,1]^m}\, |f_O(x)-\hat{f}_{O}(x)|\le \epsilon\,\sum_{i=1}^M\|O_i\|_\infty\,,
		\end{align*}
		with probability at least $1-\delta$, whenever
		\begin{align*}
			N=\Big(\frac{\gamma}{2}\Big)^{-[2(r+r_0+k_0)]^D\ell}{\log\Big(\frac{M}{\delta}\Big)+[2(r+r_0+k_0)]^D\ell\log\Big(\frac{2}{\gamma}\Big)}\Big(\frac{\gamma}{2}\Big)^{-[2(r+r_0+k_0)]^D\ell}
		\end{align*}
		with
		\begin{align*}
			&  r=\left\lceil 2\xi\log\left(\frac{16 \beta (C+c')\,h\,(2r_0+k_0)^D (D-1)!(2\xi)^{D-1}\,D^{D-1}}{\epsilon\, e^{\frac{k_0+1}{2\xi}}(1-e^{-\frac{1}{2\xi}})}\right)\right\rceil\,,\\
			& \gamma=\frac{\epsilon\,e^{-[2(r+k_0)]^D(3\log 2+5\beta h)}}{2[2(r+k_0)]^Dh\ell}\,.
		\end{align*}
	\end{proposition}
	Before we prove this result, let us simplify its statement. First, $r=\Theta(\log(\epsilon^{-1}))$, $\gamma=\Theta\Big(\epsilon\, \frac{e^{-\log(\epsilon^{-1})^D}}{\log(\epsilon^{-1})^D}\Big)$, so that the number of samples needed is asymptotically
	\begin{align*}
		N=\Theta \left(\log\Big(\frac{M}{\delta}\Big)\, e^{\operatorname{polylog}(\epsilon^{-1})} \right)\,.
	\end{align*}
	
	\begin{proof}
		% A simple calculation gives that $\mathbb{P}\big(\forall j;\,\|Y_j|_{S_i(r)}-x|_{S_i(r)}\|_{\ell^\infty}>\epsilon\big)=\big(1-\epsilon^{|S_i(r)|}\big)^N$. Indeed:
		% \begin{align*}
			%   \mathbb{P}\big(\forall j\in [N];\,\|Y_j|_{S_i(r)}-x|_{S_i(r)}\|_{\ell^\infty}>\epsilon\big)&=   \prod_{j\in [N]}\,\mathbb{P}\big(\|Y_j|_{S_i(r)}-x|_{S_i(r)}\|_{\ell^\infty}>\epsilon\big)\\
			%   &= \mathbb{P}\big(\exists k\in S_i(r): |U_k-x_k|>\epsilon\big)^N \text{\textcolor{red}{What is $U_k$ here? If $U_k$ is an element of $Y_j$, why is there no constant in the exponent?}} \\
			%   &=\Big(1-\mathbb{P}\big(\forall k\in S_i(r): U_k\in[x_k- \epsilon,x_k+\epsilon]\big)\Big)^N   \\
			%   &=\big(1-\epsilon^{|S_i(r)|}\big)^N\,.
			% \end{align*} 
		We fix $O_i$ and $r>0$, and restrict ourselves to the subset of parameters $x|_{\mathcal{S}_i(r)}$. The number of parameters in that subset is bounded by the volume $V(r+r_0+k_0)$ of the ball $S_i(r+r_0)$ times the number $\ell$ of parameters per interaction. We denote it by $m_r:=V(r+r_0+k_0)\ell$. Next, we partition the parameter space $[-1,1]^{m_r}$ onto cubes of side-size $\gamma\in (0,1)$. 
		By the coupon collector's problem, we have that the probability that none of the sub-vectors $Y_j|_{\mathcal{S}_i(r)}$ is within one of those cubes is upper bounded by $e^{-N(\gamma/2)^{m_r}+m_r\log(2/\gamma)}$.
		By union bound, the probability that for any $i\in[M]$, any cube is visited by at least one sub-vector $Y_j|_{\mathcal{S}_i(r)}$ is lower bounded by $1-\delta$, $\delta:=Me^{-N(\gamma/2)^{m_r}+m_r\log(2/\gamma)}$. 
		In other words, with probability $1-\delta$ there is a $\hat{Y}_i(x)|_{\mathcal{S}_i(r)}$ in the $N$ samples satisfying 
		\begin{align}\label{xminusYSir}
			\|x|_{\mathcal{S}_i(r)}-\hat{Y}_i(x)|_{\mathcal{S}_i(r)}\|_{\ell_\infty}\le \gamma
		\end{align}
		for all $i\in[M]$.

		Denoting $\hat{f}_{O_i}(x):=\tr\big[O_i\,\sigma(\beta,\hat{Y}_i(x)|_{\mathcal{S}_i(r)})\big]$, we next control $|f_{O_i}(x|_{\mathcal{S}_i(r)})-\hat{f}_{O_i}(x)|$. 
		This is easily done in terms of the $\ell_\infty$ norm distance between $\hat{Y}_i(x)|_{\mathcal{S}_i(r)}$ and $x|_{\mathcal{S}_i(r)}$ by Lipschitz estimate: writing $H(x|_{\mathcal{S}_i(r)})\equiv H_1$ and $H(\hat{Y}_i(x)|_{\mathcal{S}_i(r)})\equiv H_2$, 
		\begin{align*}
			\big|f_{O_i}(x|_{\mathcal{S}_i(r)})-\hat{f}_{O_i}(x)\big|&\le \|O_i\|_\infty\,\|\sigma(\beta,x|_{\mathcal{S}_i(r)})-\sigma(\beta,\hat{Y}_i(x)|_{\mathcal{S}_i(r)})\|_1\\
			&\le \|O_i\|_\infty\,\left\|\frac{e^{-\beta H_1}}{\tr[e^{-\beta H_1}]}-\frac{e^{-\beta H_2}}{\tr[e^{-\beta H_2}]}\right\|_1\leq 2\|O_i\|_\infty(e^{\|H_2-H_1\|_\infty}-1),
		\end{align*}
		by \Cref{lem:Gibbs_perturbation}.
		We can then further bound the operator norm in the exponential using the fact that $H_1-H_2$ is only supported on $S_i(r)$ to obtain:
		\begin{align*}
			\big|f_{O_i}(x|_{\mathcal{S}_i(r)})-\hat{f}_{O_i}(x)\big|\le 2\|O_i\|_{\infty}(e^{2\beta V(r+k_0)h\gamma}-1)
			.
		\end{align*}
		where the last line comes from \Cref{xminusYSir}. 
		Combining this with \Cref{Lemma:Local_Perturbation} and local indistinguishability,
		from which we have that 
		\begin{align*}
			|\tr\big[O_i\,(\sigma(\beta,x)-\sigma(\beta,x|_{\mathcal{S}_i(r)}))\big]|,\,    |\tr\big[O_i\,(\sigma(\beta,\hat{Y}_i(x))-\sigma(\beta,\hat{Y}_i(x)|_{\mathcal{S}_i(r)}))\big]|\le C_1\,e^{-\frac{r}{2\xi}}\,\|O_i\|_\infty,
		\end{align*}
		we have proven that for all $x\in[-1,1]^m$,
		\begin{align}\label{boundgibbslearning}
			|f_{O_i}(x)-\hat{f}_{O_i}(x)|\le \big(2C_1e^{-\frac{r}{2\xi}}\,+C_2(r)\gamma\big)\,\|O_i\|_\infty\,,
		\end{align}
		where $C_2(r):=2^{2V(r+k_0)+1}e^{5\beta V(r+k_0)h}V(r+k_0)h\ell$. Now, the volume $V(s)$ of a ball of radius $s$ in $\Lambda$ is equal to 
		\begin{align*}
			V(s)=\sum_{a\le s} \binom{a+D-1}{D-1}\le (2s)^D\,.
		\end{align*}
		We then fix $r$ so that $2C_1e^{-r/2\xi}\le \epsilon/2$, $\gamma$ so that $C_2(r)\gamma\le 2^{2^{D+1}(r+k_0)^D+1}e^{5\beta 2^D(r+k_0)^Dh}2^D(r+k_0)^Dh\ell\le \epsilon/2$, and therefore a lower bound on $N$ arises from the constraint $$\delta:=Me^{-N(\gamma/2)^{m_r}+m_r\log(2/\gamma)}\,.$$ Namely:
		\begin{align*}
			&r=\left\lceil 2\xi\log\left(\frac{16 \beta (C+c')\,h\,(2r_0+k_0)^D (D-1)!(2\xi)^{D-1}\,D^{D-1}}{\epsilon\, e^{\frac{k_0+1}{2\xi}}(1-e^{-\frac{1}{2\xi}})}\right)\right\rceil\,,\\
			&\gamma=\frac{\epsilon\,e^{-[2(r+k_0)]^D(3\log 2+5\beta h)}}{2[2(r+k_0)]^Dh\ell}\,.
		\end{align*}
		Therefore,
		\begin{align*}
			N=\Big(\frac{\gamma}{2}\Big)^{-[2(r+r_0+k_0)]^D\ell}{\log\Big(\frac{M}{\delta}\Big)+[2(r+r_0+k_0)]^D\ell\log\Big(\frac{2}{\gamma}\Big)}\Big(\frac{\gamma}{2}\Big)^{-[2(r+r_0+k_0)]^D\ell}
		\end{align*}
		copies suffice for the approximation claimed to hold with probability $1-\delta$.
		
	\end{proof}
	% For this to be bounded by $\eta\|O_i\|_\infty^2$, we need $\epsilon=\sqrt{\eta} e^{-\mathcal{O}(\log(1/\eta)^D)}$. The probability of this happening at least

	% For the above upper bound to be $\mathcal{O}(\delta).\|O_i\|_\infty^2$, $N$ has to scale as ${\Omega}\big(\log(1/\eta).\big(\frac{6C}{\eta}\big)^{2D}\big)$. Under this condition, we therefore have found the following bound:

	At this stage, we use the shadow tomography protocol to get classical shadows $\widetilde{\sigma}(\beta,\hat{Y}_i(x))$ for each of the states $\sigma(\beta,\hat{Y}_i(x))$. Since only one copy of each $\widetilde{\sigma}(\beta,\hat{Y}_i(x))$ is available, its reconstruction is likely going to be too noisy. 
	Instead, we will use several non-identical copies $\sigma(\beta,\hat{Y}_i(x))$ which almost coincide on large enough regions surrounding the supports of observables $O_i$ in order to improve the precision of the estimation of $\hat{f}_O$. We first develop the following Gibbs shadow tomography protocol which we believe to be of independent interest.
	
	Consider a Gibbs state $\sigma(\beta,x)$ and a family $\sigma(\beta,x_1),\dots,\sigma(\beta ,x_N)$ of Gibbs states with the promise that for any $i\in [M]$ there exist $t$ vectors $x_{i_1},\dots, x_{i_t}$ such that $\max_{j\in[t]}\|x|_{\mathcal{S}_i(r)}-x_{i_j}|_{\mathcal{S}_i(r)}\|_\infty \le \gamma$. 
	We run the shadow protocol and construct product operators $\widetilde{\sigma}(\beta, x_1),\dots,\widetilde{\sigma}(\beta,x_N)$. 
	Then for any ball $B$ of radius $k_0$, we select the shadows $\widetilde{\sigma}(\beta,x_{i_1}),\dots \widetilde{\sigma}(\beta,x_{i_t})$ and construct the empirical average
	\begin{align*}
		\widetilde{\sigma}_{B}(x):=\frac{1}{t}\sum_{j=1}^t\,\tr_{B^c}\big[\widetilde{\sigma}(\beta,x_{i_j})\big]\,.
	\end{align*}
	
	\begin{proposition}[Robust shadow tomography for Gibbs states]\label{propgibbsshadow}
		Fix $\epsilon,\delta\in(0,1)$. In the notations of \Cref{approxbetax}, with probability $1-\delta'$, for any ball $B$ of radius $k_0$, the shadow $\widetilde{\sigma}_{B}$ satisfies $\|\widetilde{\sigma}_{B}-\tr_{B^c}[\sigma(\beta,x)]\|_1\le 2C_1\,e^{-\frac{r}{2\xi}}+C_2(r)\gamma+\epsilon$  as long as
		\begin{align}\label{ttimes}
			t\ge \frac{8.12^{k_0}}{3.\epsilon^2}\log\left(\frac{n^{k_0}2^{k_0+1}}{\delta'}\right)\,.
		\end{align}
	\end{proposition}

	\begin{proof}
		In view of \Cref{propshadow}, it is enough to show that the reduced states $\tr_{B^c}[\sigma(\beta,x_{i_j})]$ are close to $\tr_{B^c}[\sigma(\beta,x)]$. This is done by simply adapting some of the estimates in the proof of \Cref{approxbetax}. In particular, we have shown that
		\begin{align*}
			\|\tr_{B^c}\big[\sigma(\beta,x)-\sigma(\beta,x_{i_j})\big]\|_1\le 2C_1\,e^{-\frac{r}{2\xi}}+C_2(r)\gamma\,.
		\end{align*}
		The result follows directly from \Cref{propshadow}.
	\end{proof}

	We are now ready to state and proof the main result of this section. We denote $\widetilde{f}_{O_i}(x)=\tr\big[O_i\, \widetilde{\sigma}_{S_i}(x)\big]$ the function constructed from the Gibbs shadow tomography protocol of \Cref{propgibbsshadow}, and write $\widetilde{f}_O:=\sum_{i=1}^M\widetilde{f}_{O_i}$. 
	
	\begin{theorem}[Learning algorithm for quantum Gibbs states]\label{learningalgoGibbs}
		In the notation of the previous paragraph, consider a set of $N$ shadows $\{\widetilde{\sigma}(\beta,x_i) \}_{i=1}^N$.
		Given an arbitrary local observable $O$, we fix
		\begin{align*}
			&r:=\left\lceil 2\xi\log\left(\frac{24\beta (C+c')\,h\,(2r_0+k_0)^D (D-1)!(2\xi)^{D-1}\,D^{D-1}}{\epsilon\, e^{\frac{k_0+1}{2\xi}}(1-e^{-\frac{1}{2\xi}})}\right)\right\rceil,\\
			& \gamma=\frac{\epsilon\,e^{-[2(r+k_0)]^D(3\log 2+5\beta h)}}{3[2(r+k_0)]^Dh\ell},\\
			&t:=\left\lceil\frac{24.12^{k_0}}{\epsilon^2}\log\left(\frac{n^{k_0}2^{k_0+1}}{\delta'}\right)\right\rceil\,.
		\end{align*}
		Then, we have that with probability $(1-\delta).(1-\delta')$, 
		\begin{align*}
			|f_{O}(x)-\widetilde{f}_O(x)|\le  \epsilon\,\sum_{i}\,\|O_i\|_\infty\,,
		\end{align*}
		as long as 
		\begin{align*}
			N=t\Big(\frac{\gamma}{2}\Big)^{-[2(r+r_0+k_0)]^D\ell}{\log\Big(\frac{M}{\delta}\Big)+t\log\Big[t\Big(\frac{2}{\gamma}\Big)^{[2(r+r_0+k_0)]^D\ell}\Big]}\Big(\frac{\gamma}{2}\Big)^{-[2(r+r_0+k_0)]^D\ell}\,.\end{align*}
	\end{theorem}
	
	Once again, upon careful checking of the bounds, we have found 
	
	\begin{align*}
		N=\Theta \left(\log\Big(\frac{M}{\delta}\Big)\,\log\Big(\frac{n}{\delta'}\Big) e^{\operatorname{polylog}(\epsilon^{-1})} \right)\,.
	\end{align*}
	
	\begin{proof}
		Adapting the proof of \Cref{approxbetax}, it is clear that with probability  $$1-\delta:=1-Me^{-N\frac{1}{t}(\gamma/2)^{m_r}+m_r\log(2/\gamma)+\log t}$$
		each cube is visited at least $t$ times. Conditioned on that event, and choosing $t$ such that \Cref{ttimes} holds, we have that with probability $1-\delta'$ 
		\begin{align*}
			|f_{O_i}(x)-\widetilde{f}_{O_i}(x)|\le \Big(2C_1 e^{-\frac{r}{2\xi}}+C_2(r)\gamma +\epsilon \Big)\|O_i\|_\infty\,.
		\end{align*}
	\end{proof}
	
	\begin{remark}
		We emphasise that the classical data $\{\widetilde{\sigma}(\beta,x_i)\}_{i=1}^N$ are fixed, and then any local observable $O$ can be chosen \textit{after the data has been taken} which will satisfy the bounds in \cref{learningalgoGibbs}.
	\end{remark}
	\begin{remark}
		Our proof readily extends to distributions $\mu$ that satisfy the following anti-concentration property: for any $x_0\in[-1,1]^{m_r}$ and for all $\pi$ permutations of the coordinates of $[-1,1]^n$ we have:
		\begin{align}\label{equ:anticoncent}
			&\mu(A(x_0,\epsilon,\pi))>0\implies \mu(A(x_0,\epsilon,\pi))=\Omega((2\epsilon)^{m_r}/\textrm{polylog}(n))\,,
		\end{align}
		where $A(x_0,\epsilon,\pi):=\pi((x_0+[-\epsilon,\epsilon]^{m_r})\times [-1,1]^{n-m_r})$. To see this, notice that the condition in \Cref{equ:anticoncent} implies that we can discretise the number of cubes with size $\epsilon$ and positive weight into at most $\mathcal{O}((2\epsilon)^{-m_r}\textrm{polylog}(n))$ cubes for any choice of $m_r$ coordinates. By e.g. rejection sampling we can then generate a sample from the uniform distribution on those cubes by taking at most $\mathcal{O}(\textrm{polylog}(n))$ samples from the distribution $\mu$. Once given uniform samples over those cubes we can argue as in the proof above.
		
		One distribution that satisfies the condition in \Cref{equ:anticoncent} but is far from uniform over the whole space is e.g. a Dirac measure on a single state. It is also satisfied for various natural distributions, such as i.i.d. distributions on each coordinate.
	\end{remark}
	
	\begin{remark}
		It is clear that our stronger $L_\infty$  recovery guarantee cannot hold for arbitrary distributions and requires some sort of anti-concentration. To see this, consider a distribution over parameters that outputs a state $\rho_0$ with probability $1-p$ and a different, locally distinguishable state $\rho_1$ with probability $p$. Before we have drawn $\Omega(p^{-1})$ samples it is unlikely that we gained access to even a single sample of $\rho_1$. But algorithms like ours with $L_\infty$ guarantees also need to perform well on such rare outputs. Thus, we see that the sample complexity for $L_\infty$ guarantees will have to depend on the parameter $p$ and will blow-up as $p\to0$. In contrast, if we wish to obtain good recovery in $L_2$ for this simple example as $p\to0$, we can always output the expectation value w.r.t. $\rho_0$.
	\end{remark}

	\begin{remark}
		Once we are given the support of the observable $O_i$, it is straightforward to generate a data structure that allows for a highly efficient evaluation of $f_O(x)$ on some parameters $x$. Assume w.l.o.g. that $S_i(r)$ consists of the the first $l$ coordinates of $x$. Then we can sort all of the sampled parameters and insert them into a search tree in time $\mathcal{O}(N\log(N))$. Afterwards, we can find all elements $y$ s.t. $\|x_{S_i(r)}-y_{S_i(r)}\|_{l\infty}\leq\epsilon$ in time
		\begin{align}
			\mathcal{O}(t\log(N))=\mathcal{O}\left( \frac{24.12^{k_0}}{\epsilon^2}\log\left(\frac{n^{k_0}2^{k_0+1}}{\delta'}\right)\,\log(N)\right),
		\end{align}
		as we expect $t$ many elements of the list to be close to our new parameter. After we found the samples that are close, we only need  to evaluate the value of $O$ on the corresponding shadow, which gives an additional time complexity of $\mathcal{O}(t4^k_0)$, as $O_i$ has a support of size $k_0$. Thus, we see that the time complexity to learn an efficient representation of the function $f_{O_i}$ is comparable to the one in~\cite{Lewis_Huang_Preskill_2023}.
	\end{remark}

	% \subsection{Learning the High-Temperature Phase}

	\begin{theorem}[Learnability of the High-Temperature Phase] \label{Theorem:1D_Commute_Analytic}
		Let $H(x)$ be a geometrically local Hamiltonian.
		Then there exists a temperature range $\beta \in [0,\beta_c]$, such for all $x\in[-1,1]^m$
		then the parameters of \Cref{learningalgoGibbs} are sufficient to  learn
		\begin{align*}
			|f_{O}(x)-\widetilde{f}_O(x)|\le  \epsilon\,\sum_{i}\,\|O_i\|_\infty\,,
		\end{align*}
		with probabilities and parameters as given in \Cref{learningalgoGibbs}.
	\end{theorem}
	\begin{proof}
		From \cite{kliesch2014locality}, we see that for sufficiently high-temperatures (low $\beta$), then the Gibbs states must satisfy exponentially decaying correlations.
		Thus we can utilise \Cref{learningalgoGibbs} directly to get the parameters required to learn the high-temperature phase.
	\end{proof}
	
	\begin{remark}
		%  For commuting Hamiltonians, we can more properly relate to the free energy
		For 1D, translationally invariant Hamiltonians, the Gibbs state has exponential decay of correlations for all temperatures \cite{Bluhm2022exponentialdecayof} and hence the phase can be learned efficiently everywhere. 
	\end{remark}
	
	\begin{remark}
		%  For commuting Hamiltonians, we can more properly relate to the free energy
		For commuting and 1D Hamiltonians we can relate the learnability of the phase to the analyticity of the free energy, and thus to a more rigorous notion of phase, defined as regions of parameter space where the free energy is analytic (assuming the Hamiltonian is parameterised in an analytic fashion). 
		The free energy is defined as $F(\beta,x) = -\log(\tr[e^{-\beta H(x)}])$.
		This is done using \cite[Theorem 32]{harrow2020classical} which demonstrates that for commuting and 1D geometrically local Hamiltonians, exponential decay of correlations holds in the sense of \Cref{decaycorrintro}.
		Thus we can utilise \Cref{learningalgoGibbs} directly to get the parameters.

		One might attempt to relate this to the free energy of non-commuting Hamiltonians, however, as per \cite[Theorem 31]{harrow2020classical}, exponential decay of correlations in regions with analytic free energy is only known for observables $O_1, O_2$ whose supports are distance $\Omega(\log(n))$ away from each other.
		Although one would be able to prove  learnability with more samples scaling with $n$, \cite[Theorem 31]{harrow2020classical} is not strong enough to give the scaling we desire.
	\end{remark}
	\begin{remark}
		Although we do not prove it here, it is likely we can flip the above remark on its head. 
		If we consider a region of parameter space in which the free energy is analytic, we expect all local observables to be analytic in $x$ in this region.
		As such, we should be able to approximate the local observable using polynomial interpolation (or some other technique) and learning the polynomial of this observable everywhere in the phase with small error.
	\end{remark}

	\iffalse
	\cite[Theorem 32]{harrow2020classical} demonstrates that for commuting and 1D geometrically local Hamiltonians, exponential decay of correlations holds in the sense of \Cref{decaycorrintro}.
	Thus we can utilise \Cref{learningalgoGibbs} directly to get the parameters.
	\JDW{Need to include Condition 1 of Theorem 32 of \cite{harrow2020classical}.}\textcolor{blue}{CR: do we? for 1D, I'd cite Araki \cite{araki1969gibbs} and for any D high T I'd cite \cite{kliesch2014locality}}
	\fi

	It is worth noting that high-temperature Gibbs states of commuting Hamiltonians as well as those of 1D Hamiltonians can be efficiently prepared and hence desired observables could be measured directly \cite{Brandao_Kastoryano_2018}.
	Hence \Cref{Theorem:1D_Commute_Analytic} becomes most useful in the setting where parameters are a priori unknown to us and we wish to extract useful information. However, we will now generalize to computationally harder phases.

	%We note that this theorem is perhaps not as useful as it appears: it is known that Gibbs states of commuting and 1D Hamiltonians can be efficiently prepared, and hence desired observables could be prepared and measured directly \cite{Brandao_Kastoryano_2018}. \textcolor{blue}{CR: I am not sure I understand this comment. How can we prepare the state if we don't know it?}
	
	\iffalse
	One might attempt to relate this to the free energy of non-commuting Hamiltonians, however, as per \cite[Theorem 31]{harrow2020classical}, exponential decay of correlations in regions with analytic free energy is only known for observables $O_1, O_2$ whose supports are distance $\Omega(\log(n))$ away from each other.
	However, \Cref{learningalgoGibbs} does not straightforwardly apply in this case.
	\JDW{Probably can modify it with $exp(\operatorname{polylog}(n/\epsilon))$ overhead if you wanted.}
	\textcolor{blue}{Is this true? Can we not leverage the results of \cite{kliesch2014locality}?} \JDW{While \cite{kliesch2014locality} does given exponential decay of correlations, it doesn't relate it to analyticity of the free energy.} \textcolor{blue}{CR: sure}
	\fi

	\section{Learning beyond high temperature phases}\label{sec:beyond_high_T}
	The previous results only hold for phases for which there exist efficient classical algorithms to compute the expectation value of local observables. Indeed, this follows directly from the fact that it is possible to compute local expectation values on some region $S$ by evaluating them on the Gibbs state $\sigma(\beta,(x|_{S(r)},0_{S^c}))$, as this Gibbs state will be maximally mixed outside of $x|_{S(r)}$ and on $S(r)$ we can just diagonalise it. It would be desirable to extend our results to phases for which it is not known how to compute local expectation values efficiently on a classical computer.
	
	However, the previous results have pedagogical value and can be easily generalized by our previous trick, namely shifting the ``center'' of the parameters. Indeed, if we assume instead that for every region $S$, there is a fixed configuration $x^*_{S(r)^c}$ s.t. the Gibbs states with parameters $x=(x_S,x^*_{S(r)^c})$ approximate well the reduced density matrix of $\sigma(\beta,x)$ on $S$, then we can prove similar results as before. But now, it is not generally the case that these states can be computed efficiently without access to samples, as $\sigma(\beta,(x_S,x^*_{S(r)^c}))$ will generally be hard to compute \cite{Sly2010}. In a nutshell, our generalisation consists of changing the parameters outside of the region from $0$ to some other fixed value.

	The central condition we need to generalise to go beyond trivial phases in our learning algorithm is that of local indistinguishability. For that, we introduce what we call \emph{generalised approximate local indistinguishability} (GALI):
	\begin{definition}[Generalised approximate local indistinguishability (GALI)]\label{def:GALI}
		For $x_1^0,\ldots,x_m^0\in\mathbb{R}$ let 
		$\Phi\coloneqq \prod_{i=1}^m [-1+x_i^0,1+x_i^0]$ and for $x\in\Phi$ define $\sigma(\beta,x)$ as before for some fixed $\beta$. We say that the family of Gibbs states $\sigma(\beta,x)$ satisfies generalised approximate local indistinguishability (GALI) with parameter $\eta(S)>0$ and decay function $f$ if
		for any region $S$ and $r$ there is a set of parameters $x^*_{\mathcal{S}(r)^c}$ s.t. for all $O$ supported on $S$ and $f_{O}(x):=\tr[O\,\sigma(\beta,x)]$ the following bound holds:
		\begin{align*}
			\sup_{x\in\Phi} |f_{O}(x)-f_{O}((x|_{\mathcal{S}(r)},x^*_{\mathcal{S}(r)^c})) | \le (|S|f(r)+\eta(S))\,\|O\|_\infty\,,
		\end{align*}
		with $\lim_{r\to\infty}f(r)=0$.
	\end{definition}
	\begin{remark}
		In the above definition we have shifted the center of $\Phi$, but have not changed the width of each side of the hypercube. However, it is simple also to change the maximal width of the parameters we consider for each local interaction by suitably rescaling the Hamiltonian by a constant.
	\end{remark}
	Note that the usual notion of local indistinguishability holds by taking the vector $x^*_{S(r)^c}$ to be zero for all regions $S$. Then we can once again set our estimator as in \Cref{equ:defin_estimator} by just considering all the points we sampled that are sufficiently close to the one we care about on an enlarged region.
	\begin{theorem}\label{thm:generalized_learning_gali}
		Assuming that GALI holds as in \Cref{def:GALI} and for a given $\epsilon>0$ let $r(\epsilon)$ be the smallest s.t. for each $O_i$ $|\operatorname{supp}(O_i)|f(r(\epsilon))\leq \epsilon$. Let $\gamma\leq (|B(\operatorname{supp}(O_i),r(\epsilon))|\ell)^{-1}\epsilon$. Furthermore, assume that we have enough samples $N$ of $y_i$ to ensure that we have $m$ samples $y_
		{i_1},\ldots,y_{i_m}$ with
		\begin{align}\label{equ:close_region}
			\|y_{i_j}|_{B(\operatorname{supp}(O_i),r(\epsilon))}-x|_{B(\operatorname{supp}(O_i),r(\epsilon))}\|_{\ell_\infty}\leq \gamma,
		\end{align}
		where $m=e^{\mathcal{O}(B(\operatorname{supp}(O_i),r(\epsilon)))}\mathcal{O}(\epsilon^{-2}\log(\delta^{-1}))$. Then the estimator $\hat{f}_O(x):=\sum_{i}\,\tr[O_i\,\sigma(\beta,\hat{Y}_i(x))]$ satisfies the bound
		\begin{align*}
			\sup_{x\in\Phi}\, |f_O(x)-\hat{f}_{O}(x)|\le \mathcal{O}\Big(\epsilon\,\sum_{i=1}^M(1+\eta(\operatorname{supp}(O_i)))\|O_i\|_\infty\Big)\,,
		\end{align*}
		with probability at least $1-\delta$.
	\end{theorem}
	\begin{proof}
		We will proceed in a similar fashion as \Cref{approxbetax}. First, note that all Gibbs states of the form $\sigma(\beta,y_i|_{B(\operatorname{supp}(O_i),r(\epsilon))},x^*_{B(\operatorname{supp}(O_i),r(\epsilon))^c})$ are close to $\sigma(\beta,x|_{B(\operatorname{supp}(O_i),r(\epsilon))},x^*_{B(\operatorname{supp}(O_i),r(\epsilon))^c})$. Indeed, by \Cref{lem:Gibbs_perturbation} we have:
		\begin{align*}
			&\|\sigma(\beta,y_i|_{B(\operatorname{supp}(O_i),r(\epsilon))},x^*_{B(\operatorname{supp}(O_i),r(\epsilon)^c})-\sigma(\beta,x|_{B(\operatorname{supp}(O_i),r(\epsilon))},x^*_{B(\operatorname{supp}(O_i),r(\epsilon)^c})\|_1\\
			&\qquad\qquad\qquad\leq 2(e^{\beta \ell|B(\operatorname{supp}(O_i),r(\epsilon))|\gamma}-1)=\mathcal{O}(\epsilon),
		\end{align*}
		by our choice of $\gamma$.
		Now, by our choice of $r(\epsilon)$  and the definition of GALI we have that 
		\begin{align*}
			&\|\tr_{B(\textrm{supp}(O_i),r(\epsilon))^c}\left[\sigma(\beta,x|_{B(\textrm{supp}(O_i),r(\epsilon))},x^*_{B(\textrm{supp}(O_i),r(\epsilon)^c})-\sigma(\beta,x)\right]\|_1\\
			&\qquad\qquad\qquad\leq |\textrm{supp}(O_i)|f(r(\epsilon)+\eta(\textrm{supp}(O_i))\leq \epsilon+\eta(\textrm{supp}(O_i)),
		\end{align*}
		and similarly for $\|\sigma(\beta,y|_{B(\textrm{supp}(O_i),r(\epsilon))},x^*_{B(\textrm{supp}(O_i),r(\epsilon))^c})-\sigma(\beta,y)\|_1$. From this we conclude by applying triangle inequalities that 
		\begin{align}
			\|\tr_{B(\textrm{supp}(O_i),r(\epsilon))^c}\left[\sigma(\beta,x)-\sigma(\beta,y_i)\right]\|_1=\mathcal{O}(\epsilon+\eta(\textrm{supp}(O_i)))
		\end{align}
		Thus, by our robust shadow protocol, all we need is to gather enough shadows from $y_i$ to recover the reduced density matrix of $\sigma(\beta,x)$ on the region $B(\textrm{supp}(O_i),r(\epsilon))$. This gives the required bound on $m$. 
	\end{proof}
	
	Let us make some remarks on the various scalings. First, note that assuming a uniform distribution over the parameters, the probability $p$ that we draw one sample from ``the right corner'' is $p=\mathcal{O}(\gamma^{|B(O_i,r(\epsilon))|}\log(\gamma^{|B(O_i,r(\epsilon))|}))$. Thus, by the coupon collector problem, we will need $\mathcal{O}(mp^{-1}\log(p^{-1}))$ samples to have collected enough samples to perform the estimate. In particular, if $f$ is exponentially decaying, we recover our previous quasi-polynomial scaling in precision.

	\subsection{Exponentially decaying phases}
	
	%\noindent The rest of the proof of \cref{learningalgoGibbs} follows by essentially the same proofs as previously. Thus, we see that our results also cover phases of Gibbs states that can be computationally non-trivial.

	%The only technical impediment for proving local indistinguishability is reproving \cref{Lemma:Local_Perturbation}.
	The first class of systems for which we will show that GALI holds is that of Gibbs states satisfying exponential decay of correlations.

	\begin{proposition}[Gibbs local indistinguishability (with a shifted parameter)] \label{Lemma:Local_Perturbation_2}
		Assume the exponential decay of correlations in \Cref{decaycorr} for the region $\Phi\coloneqq \prod_{i=1}^m [-1+x_i^0,1+x_i^0]$. Then the family of Gibbs states $\sigma(\beta,x)$ for $x\in\Phi$ satisfies GALI with $f(r)=\mathcal{O}(e^{-\tfrac{r}{\xi}})$ and $\eta=0$ and, thus, \Cref{thm:generalized_learning_gali} applies.
	\end{proposition}
	\begin{proof}
		The proof is completely analogous to that of \Cref{Lemma:Local_Perturbation}, but instead of taking a path where we set the coordinates to $0$ outside of region, we set them all to the parameters $x^*_S$ given by GALI.
	\end{proof}

	%For pedagogical reasons we will first prove our results for LTQO ground states. Note, however, that it is always efficient to compute local expectation values of ground states that satisfy LTQO: indeed, \Cref{defi:LTQO} immediately implies that to compute a local expectation value on a region $A$ all we need to do is to find the ground state of the family of Hamiltonians on a suitably enlarged region $B$ and then evaluate the expectation value. However, as before in e.g. \Cref{sec:beyond_trivial}, it will be straightforward to extend our results to phases that satisfy a "shifted" version of LTQO. That is, by taking 

	%The previous results can be adapted to ground states with exponentially decaying correlations through the following tricks. Here by ground state we mean the limit
	%\begin{align*}
	%    \psi_g(x):=\lim_{\beta\to\infty}\, \sigma(\beta,x)\,.
	%\end{align*}
	
	%For the rest of the section we will consider the same family of Hamiltonians $\{H(x)\}_{x\in [-1,1]^m}$ as the one defined in \Cref{Hamiltoniandef}.
	%We make the assumption $x\in [-1,1]^m$ for simplicity, but --- as for the Gibbs states results --- the results here will hold for any shifted phase such as $x\in \prod_{i=1}^m [-1+x_i^*, 1 +x_i^*]$.
	One example of a system for which these conditions apply are those described in~\cite{Lubetzky2013,pirogov_sinai}, where we pick the center to be far enough from the phase transition.
	It is expected that many models have exponential decay of correlations when away from criticality.
	
	\ \newline 
	\noindent \textbf{A Complexity Theoretic Speedup for Learning Gibbs Properties using ML?}
	If we wish to prove a complexity theoretic speed up, it is necessary to first prove that it is difficult to approximate the properties of Gibbs states with exponentially decaying correlations.
	At the present time of writing, this is not known.
	It is know that, in general, Gibbs state properties are hard to approximating for the complexity class \textsf{QXC}  \cite{bravyi2022quantum, bravyi2023quantum}, however, it is not known to hold for exponentially decaying correlations.
	It is not unreasonable to expect that it is computationally intractable to approximate Gibbs state properties in the presence of exponential decays given that this is true for ground state properties (as gapped ground state properties are difficult to approximate).
	Furthermore, intuitively, if we bring the temperature of the Gibbs state down to a low but constant value, the properties of the Gibbs state should be similar to the ground state.
	We expect further work in Hamiltonian complexity to elucidate this matter.

	\subsection{Ground states} \label{Sec:Learning_Ground_States}
	
	Condition \ref{def:GALI} is trivially satisfied by phases of classical Hamiltonians each possessing a unique ground state, since the latter are simply tensor products of pure computational basis states and are therefore locally indistinguishable from one another. This simple observation can in fact be extended to gapped quantum phases. The argument makes use of the spectral flow, a tool that has already been used in \cite{Lewis_Huang_Preskill_2023}. A simple adaptation of their proof allows us to show the following 
	
	\begin{proposition}[Adapted from \cite{Lewis_Huang_Preskill_2023}]
		Consider a Hamiltonian $H(x)$, which is promised to have its spectral gapped bounded from below by $\gamma$ in the region $x\in \Phi$ for $\Phi\coloneqq \prod_{i=1}^m [-1+x_i^0,1+x_i^0]$.
		Then the GALI condition defined in \cref{def:GALI} holds in this region holds for $f$ an exponentially decaying function and $\eta(S)=0$.
	\end{proposition}
	\begin{proof}
		
		We single out a single coordinate $x_k$ which we will vary and consider a local observable $A$.
		Using \cite[Corollary 2.8]{bachmann2012automorphic}, we have:
		\begin{align*}
			\frac{\partial }{\partial x_k}\tr[A\rho(x)] &= \int^{\infty}_{-\infty}W_\gamma(t)\sum_{j=1}^m\bigg[ A, e^{iH(x)t} \frac{\partial h_j(x_j)}{\partial x_k} e^{-iH(x)t} \bigg] dt \\
			&= \int^{\infty}_{-\infty}W_\gamma(t)\bigg[ A, e^{iH(x)t} \frac{\partial h_k(x_k)}{\partial x_k} e^{-iH(x)t} \bigg] dt \\
			\left|\frac{\partial }{\partial x_k}\tr[A\rho(x)] \right| &= \int^{\infty}_{-\infty}W_\gamma(t)\norm{\bigg[ A, e^{iH(x)t} \frac{\partial h_k(x_k)}{\partial x_k} e^{-iH(x)t} \bigg] }_\infty dt
		\end{align*}
		where in the second line we have used that only $h_k(x_k)$ depends on the parameter $x_k$.
		Here $W_\gamma(t)$ is a function satisfying:
		\begin{align*}
			|W_\gamma(t)| \leq \begin{cases} 
				1/2 & 0 \leq \gamma |t|\leq \theta \\
				35e^2(\gamma|t|)^4e^{-\frac{2}{7}\frac{\gamma|t|}{\log^2(\gamma |t|)}} & \gamma |t|> \theta,
			\end{cases}
		\end{align*}
		where $\theta$ is the largest real solution to $35e^2(\gamma|t|)^4e^{-\frac{2}{7}\frac{\gamma|t|}{\log^2(\gamma |t|)}} =1/2$.
		
		We can then bound this using standard Lieb-Robinson techniques by splitting up the integral:
		\begin{align*}
			& \int^{\infty}_{-\infty}W_\gamma(t)\norm{\bigg[ A, e^{iH(x)t} \frac{\partial h_k(x_k)}{\partial x_k} e^{-iH(x)t} \bigg] }_\infty dt \\
			&\qquad\qquad  = \left( \int^{t^*}_{-t^*}+ \int^{\infty}_{t^*}+ \int^{-t^*}_{-\infty} \right) W_\gamma(t)\norm{\bigg[ A, e^{iH(x)t} \frac{\partial h_k(x_k)}{\partial x_k} e^{-iH(x)t} \bigg] }_\infty dt\,.
		\end{align*}
		
		Using Lieb-Robinson techniques, such as \cite[Lemma 3]{Lewis_Huang_Preskill_2023}, one can demonstrate that: 
		\begin{align*}
			\left|\frac{\partial }{\partial x_k}\tr[A\rho(x)] \right| \leq C\norm{A}_\infty e^{-\operatorname{dist}(A, h_j)/\xi}
		\end{align*}
		for an appropriately chosen $\xi$. We now follow the proof of \cref{Lemma:Local_Perturbation}.
		We identify $x|_{S(r)}$ with the vector $(x|_{S(r)}, x^*|_{S^c(r)})$.
		Given a path $x(s) = (1-s)x + s x|_{S(r)}$, we get 
		\begin{align*}
			|  \tr{[A\rho(x)]} - \tr{[A\rho(x|_{S(r)})]}| &\leq \sum_{l\in S^c(r)}  \int_0^1 \left|\frac{\partial \tr[A\rho(x)]}{\partial x_l}  \right| ds \\
			&\leq C\norm{A}_\infty\sum_{l\in S^c(r)}  \int_0^1 e^{-\operatorname{dist}(A, h_l)/\xi} ds \\
			&\leq C\norm{A}_\infty \sum_{l\in S^c(r)}e^{-\operatorname{dist}(A, h_l)/\xi}
		\end{align*}
		By the same methods as \cref{Lemma:Local_Perturbation}, we can bound this as:
		\begin{align*}
			|  \tr{[A\rho(x)]} - \tr{[A\rho(x|_{S(r)})]}|     &\leq C\norm{A}_\infty (D-1)!(2\xi)^{D-1}D^{D-1}\frac{e^{-(r+k_0/2 + 1)/2\xi}}{1 - e^{1/2\xi}} \\
			&\equiv  C'\norm{A}_\infty e^{-r/2\xi}.
		\end{align*}
		
		\noindent Thus gapped ground states satisfy the GALI condition, \Cref{def:GALI}.
		
	\end{proof}
	
	\begin{remark}
		
		We note, however, that ground states of general Hamiltonians do not necessary satisfy \Cref{def:GALI}.
		For example, the 1D, nearest neighbour Hamiltonian in \cite{gottesman2009quantum} changes phases as the size of the 1D chain increases, and thus does not satisfy GALI (\Cref{def:GALI}).
		
	\end{remark}
	
	\begin{remark}
		
		There are also Hamiltonians which we expect GALI (\Cref{def:GALI}) to hold for some sufficiently large lattice size, $L\geq L_0$, but where $L_0$ is uncomputable.
		Examples include the Hamiltonians specified in \cite{Cubitt_Perez-Garcia_Wolf_2015, bausch2020undecidability, bausch2021uncomputability}. 
		The same set of Hamiltonians can be used to prove that determining whether GALI is satisfied at all is itself an uncomputable property as $L\rightarrow \infty$.
		
	\end{remark}

	\subsection{Low temperature Gibbs distributions}

	In \cite{Lewis_Huang_Preskill_2023}, the authors provided a class of classical Hamiltonians whose ground state properties are computationally hard to estimate, whereas  these models can be efficiently learned from few copies. More generally, the task of efficiently approximating low temperature properties of classical systems is known to be hard in general \cite{Sly2010}, and finding efficient algorithms for specific models still constitutes an active area of research in theoretical computer science \cite{pirogov_sinai}.

	Here instead, we show that low temperature Gibbs states can still be approximately learned with error scaling with temperature. More precisely, we show that low temperature Gibbs distributions with a unique ground state can be learned efficiently with an error depending on the temperature of the state. To prove this claim, we start once again from the trivial observation that the ground state of such a system can be learned with access to a unique copy of the state. Indeed, that state coincides with the Kronecker delta at the bit-string minimising the total energy of the system. 
	
	Here, we adopt the following notations for classical Gibbs distributions on bit-strings: the Hamiltonian $H$ is a function $H:\mathbb{Z}_2^{\Lambda}\to \mathbb{R}$, and the corresponding Gibbs distribution at the inverser temperature $\beta$ is denoted by $\sigma^\beta( \alpha)=e^{-\beta H(\alpha)}{\big(\sum_{\alpha'\in\mathbb{Z}_2^\Lambda}e^{-\beta H(\alpha')}\big)^{-1}}$ for all $\alpha\in\mathbb{Z}_2^\Lambda$. The unique  bit-string minimising $H$ is denoted as $\alpha^*$. Given a distribution $\mu$ on bit-strings of length $|\Lambda|$ and a region $A\subset \Lambda$, we denote by $\mu_A(\alpha_A|\alpha_{A^c})$ the measure on bitstrings of length $|A|$ conditioned on the value of the spins outside of $A$ being fixed as $\alpha_{A^c}\in\mathbb{Z}_2^{A^c}$, and evaluated at the bit-string $\alpha_A\in\mathbb{Z}_2^A$. We also denote by $\mu_A$ the marginal onto region $A$. The Kronecker delta at $\alpha^*$ is denoted by $\delta^{\alpha^*}$. The total variation between $\delta_A^{\alpha^*}$ and $\sigma_A^{\beta}$ can be expressed as
	
	\begin{align*}
		\|\delta_A^{\alpha^*}-\sigma_A^\beta\|_{\operatorname{TV}}&= \sum_{\alpha_A} \Big| \sum_{\alpha_{A^c}} \delta^{\alpha^*}(\alpha_A,\alpha_{A^c})-\sigma^\beta(\alpha_A,\alpha_{A^c}) \Big|\\
		&=\sum_{\alpha_A}\,\Big| \delta^{\alpha^*}(\alpha_A,\alpha^*_{A^c})-\sum_{\alpha_{A^c}}\sigma^\beta(\alpha_A,\alpha_{A^c}) \Big|\\
		&=\sum_{\alpha_A\ne \alpha_A^*}\,\sum_{\alpha_{A^c}} \sigma^\beta(\alpha_A,\alpha_{A^c}) +\Big|1-\sum_{\alpha_{A^c}}\sigma^\beta(\alpha_A^*,\alpha_{A^c})\Big|\\
		&=1-\sum_{\alpha_{A^c}}\sigma^\beta(\alpha_A^*,\alpha_{A^c})+\Big|1-\sum_{\alpha_{A^c}}\sigma^\beta(\alpha_A^*,\alpha_{A^c})\Big|\\
		&=2\Big[1-\sum_{\alpha_{A^c}}\sigma^\beta(\alpha_A^*,\alpha_{A^c})\Big]\\
		&=2\big[1-\sigma_A^\beta(\alpha_A^*) \big]\,.
	\end{align*}
	
	Since $\sigma^\beta$ is a Markov random field by the Hammersley-Clifford theorem \cite{clifford1971markov}:
	\begin{align*}
		\sigma^\beta_A(\alpha_A^*)&=\sum_{\alpha_{A^c}} \sigma^\beta(\alpha_A^*,\alpha_{A^c})\\
		&=\sum_{\alpha_{A^c}} \sigma^\beta(\alpha_A^*|\alpha_{A^c}) \,\sigma^\beta_{A^c}(\alpha_{A^8})\\
		&=\sum_{\alpha_{A^c}} \sigma^\beta(\alpha_A^*|\alpha_{\partial A}) \sigma^\beta_{A^c}(\alpha_{A^c})\\
		&\ge \min_{\alpha_{\partial A}}\,\sigma^\beta(\alpha_A^*|\alpha_{\partial A})\,,
	\end{align*}
	where $\alpha_{\partial A}$ denotes the spin configurations at the boundary of $A$. Plugging this bound in the total variation above, we have that
	\begin{align*}
		\|\delta_A^{\alpha^*}-\sigma^\beta_A\|_{\operatorname{TV}}\le 2\big[1-\min_{\alpha_{\partial A}}\sigma^\beta(\alpha_A^*|\alpha_{\partial A})\big]\,.
	\end{align*}
	
	With this bound, we directly arrive at the following result by simple approximation of Gibbs states by ground states:
	\begin{proposition}
		Gibbs states of classical Hamiltonians with unique ground state satisfy GALI with $\xi=0$ and $\eta= 4\max_{\alpha_{\partial S}}\big[1-\sigma^\beta(\alpha_S^*|\alpha_{\partial S})\big]$.
		
	\end{proposition}

	%\textcolor{red}{remark on extension to quantum? can recycle perhaps some of the things we wrote in the previous ground state section?}

	\section{Non-Linear Parameterisations of the Hamiltonian}\label{nonlinearparam}
	Here we show that taking a parameterisation of the Hamiltonian that is not a sum of Paulis does not change the results of \Cref{seclearningalgogibbs}.

	\begin{lemma} \label{Lemma:Non-Linear_Parameterisation}
		Consider a Hamiltonian parameterised in terms of Pauli strings:
		\begin{align*}
			H(x) = \sum_j x_j P_j
		\end{align*}
		where $P_j$ is a Pauli string.
		Consider an alternative parameterisation of the same Hamiltonian in terms of the local terms:
		\begin{align*}
			H(y) = \sum_j h_j(y^{(j)})
		\end{align*}
		where $y^{(j)}\in [-1,1]^b$, $b=O(1)$, and  $h_j$ only depends on the coordinates in $y^{(j)}$.
		We will also assume that each $h_j$ is $k$-local and $\norm{\partial_{u}h_j(y)}\leq 1$.
		
		Then, assuming all the non-zero elements of Jacobian are bounded as $1/C'\leq | \partial y_m/\partial x_k| \leq C $ for $C,C'=O(1)$, the following holds:
		\begin{align*}
			\norm{  \frac{\partial h_j(y^{(j)})}{\partial x_m} } &\leq  b4^k C,
		\end{align*}
		and 
		\begin{align*}
			|\partial_{y_m} f_L(y)| &\leq  C''\max_m  | 
			\partial_{x_i} f_L(\beta,x) |,
		\end{align*}
		where $C'=O(1)$.
	\end{lemma}

	\begin{proof}
		
		We see that:
		\begin{align*}
			\frac{\partial h_j(y^{(j)})}{\partial x_m} = \sum_i\frac{\partial y_i}{\partial x_m} \frac{\partial h_j(y^{(j)})}{\partial y_i}
		\end{align*}
		We note that $\frac{\partial h_j(y^{(j)})}{\partial y_i}$ is only non-zero for a $b$ terms.
		Furthermore, since $h_j$ is $k$-local, then it can be written as a sum of $\leq 4^k$ Pauli strings.
		Hence:
		\begin{align*}
			\norm{  \frac{\partial h_j(y^{(j)})}{\partial x_m} }&\leq \sum_i \Big| \frac{\partial y_i}{\partial x_m} \Big| \norm{ \frac{\partial h_j(y^{(j)})}{\partial y_i} } \\
			&\leq b4^kC \max \norm{ \frac{\partial h_j(y^{(j)})}{\partial y_i} }  \\
			&\leq  b4^k C= O(1).
		\end{align*}
		
		We now consider the functions $f_L(y) = \tr[L\rho(y)]$, $f_L(\beta,x) = \tr[L\rho(x)]$.
		\begin{align*}
			\partial_{y_m} f_L(y) = \sum_i \frac{\partial x_i}{\partial y_m} \partial_{x_i} f_L(\beta,x)
		\end{align*}
		Using that a given $y_m$ can depend on at most $4^k$ $x_m$ coordinates, we see that for a given $y_m$, at most $\operatorname{poly}(4^k)=O(1)$ many $ \frac{\partial x_i}{\partial y_m}$ can be non-zero.
		Thus
		\begin{align*}
			|\partial_{y_m} f_L(y)| &\leq \operatorname{poly}(4^k) \Big|\frac{\partial x_i}{\partial y_m}\Big| |\partial_{x_i} f_L(\beta,x)| \\
			&\leq \operatorname{poly}(4^k)C' | \partial_{x_i} f_L(\beta,x)|.
		\end{align*}
		
		The lemma statement then follows for $C''=\operatorname{poly}(4^k)C'$.
	\end{proof}
	
	This lemma allows us to prove up bounds on the derivative of $f_L(y)$, and thus an equivalent to \Cref{LemmaapproxPhi} holds for local observables.
	The rest of the results in \cref{seclearningalgogibbs} follow similarly.

	\iffalse
	We now consider the functions $f_L(y) = \tr[L\rho(y)]$, $f_L(\beta,x) = \tr[L\rho(x)]$.
	\begin{align*}
		\partial_{x_i} f_L(\beta,x) = \sum_m \frac{\partial y_m}{\partial x_i} \partial_{y_m} f_L(y)
	\end{align*}
	Again using that a given $y_m$ can depend on at most $4^k$ $x_m$ coordinates, we see that  
	\begin{align*}
		|\partial_{x_i} f_L(\beta,x)| &\leq  \sum_m |\frac{\partial y_m}{\partial x_i}| |\partial_{y_m} f_L(y)| \\
		&\leq 4^k \max_m | \partial_{y_m} f_L(y)|.
	\end{align*}
	
	\fi
	
	\section{Shadow tomography for non-identical copies}
	
	In this appendix, we extend the shadow tomography protocol to the case of non-identical copies. Consider a state $\sigma$ and a family $\sigma_1,\dots,\sigma_N$ of states over $n$ qubits with the promise that for any subset $A$ of qubits of size $|A|\le r$ there exists a subfamily of states $\sigma_{i_1}\dots \sigma_{i_t}$, flagged in advance, with the promise that $\max_{j\in[t]}\|\tr_{A^c}(\sigma_{i_j}-\sigma)\|_1\le \eta$.
	
	We run the shadow protocol and construct product operators $\widetilde{\sigma}_1,\dots,\widetilde{\sigma}_N$. Then for any region $A$, we select the shadows $\widetilde{\sigma}_{i_1},\dots \widetilde{\sigma}_{i_t}$ and construct the empirical average
	\begin{align*}
		\widetilde{\sigma}_A:=\frac{1}{t}\sum_{j=1}^t\,\tr_{A^c}(\widetilde{\sigma}_{i_j})\,.
	\end{align*}

	\begin{proposition}[Shadows for non-identical copies]\label{propshadow}
		Fix $\epsilon,\delta\in(0,1)$. With probability $1-\delta$, the shadow $\widetilde{\sigma}_A$ satisfies $\|\widetilde{\sigma}_A-\sigma_A\|_1\le \epsilon+\eta$ as long as
		\begin{align*}
			t\ge \frac{8.12^r}{3.\epsilon^2}\log\left(\frac{n^r 2^{r+1}}{\delta}\right)\,.
		\end{align*}
		
	\end{proposition}
	
	In order to prove the above proposition, we need an extension of the matrix Bernstein inequality used in proving the convergence guarantee of the standard shadow protocol to the case of independent, non-identically distributed random matrices:

	\begin{lemma}[Matrix Bernstein for non-i.i.d. random matrices \cite{tropp2015introduction}]\label{matrixBern}
		Let $S_1,\dots ,S_t$ be independent, centered random matrices with
		common dimension $d_1\times d_2$, and assume that each one is uniformly bounded:
		\begin{align*}
			\mathbb{E}[S_j]=0\qquad \text{ and }\qquad \|S_j\|_\infty\le L\qquad \text{ for all $j=1,\dots, t$.}
		\end{align*}
		Denote the sum $Z=\sum_{j=1}^t S_j$ and let $\nu(Z):=\max\big\{\|\mathbb{E}[ZZ^*]\|_\infty,\,\|\mathbb{E}[Z^*Z]\|_\infty\big\}$. Then,
		\begin{align*}
			\mathbb{P}\big(\|Z\|_\infty\ge s\big)\le (d_1+d_2)\,\operatorname{exp}\left(\frac{-s^2/2}{\nu(Z)+Ls/3}\right)\,\qquad \text{ for all $s\ge 0$}\,.
		\end{align*}
	\end{lemma}

	\begin{proof}[Proof of \Cref{propshadow}]
		In the notations of the previous paragraph and of \Cref{matrixBern}, we take $S_j:=\tr_{A^c}(\widetilde{\sigma}_{i_j}-\sigma_{i_j})$, $j=1\dots t$, so that $Z/t=\widetilde{\sigma}_A-\mathbb{E}[\widetilde{\sigma}_A]$. Adapting the proof for the standard shadow tomography protocol (see e.g. \cite{huang2021provably}), we have 
		\begin{align*}
			\|\tr_{A^c}(\widetilde{\sigma}_{i_j}-\sigma_{i_j})\|_\infty \le 2^r+1\qquad \text{ and }\qquad \frac{\nu(Z)}{t}=\frac{1}{t}\,\Big\|\sum_{j=1}^t\,\mathbb{E}[S_j^2]\Big\|_\infty\le 3^r\,.
		\end{align*}
		Since $\|X\|_\infty\le \|X\|_1\le 2^r\,\|X\|_\infty$, we have
		
		\begin{align*}
			\mathbb{P}\Big(\|\widetilde{\sigma}_A-\mathbb{E}[\widetilde{\sigma}_A]\|_1\ge s\Big)\le 2^{r+1}\,\operatorname{exp}\left(\frac{-ts^2/2^{2r+1}}{3^r+(2^r+1)s/(3.2^r)}\right)\le 2^{r+1}\,\operatorname{exp}\Big(\frac{-3ts^2}{8.12^r}\Big)\,.
		\end{align*}
		Next, we observe that under the assumption $\max_{j\in[t]}\|\tr_{A^c}(\sigma_{i_j}-\sigma)\|_1\le \eta$, then:
		
		\begin{align*}
			\mathbb{E}[\widetilde{\sigma}_A]=\frac{1}{t}\sum_{j=1}^t\, \tr_{A^c}(\sigma_{i_j})\qquad \Rightarrow \qquad \|\mathbb{E}[\widetilde{\sigma}_A]-\tr_{A^c}(\sigma)]\|_1\le \eta\,.
		\end{align*}
		Hence, 
		\begin{align*}
			\mathbb{P}\Big(\|\widetilde{\sigma}_A-\tr_{A^c}(\sigma)\|_1\ge \eta+\epsilon\Big)\le \mathbb{P}\Big(\|\widetilde{\sigma}_A-\mathbb{E}[\widetilde{\sigma}_A]\|_1\ge s\Big)\le 2^{r+1}\,\operatorname{exp}\Big(\frac{-3t\epsilon^2}{8.12^r}\Big)\,.
		\end{align*}
		By union bound, the result follows after choosing $\delta:=n^r2^{r+1}\,\operatorname{exp}\Big(\frac{-3t\epsilon^2}{8.12^r}\Big)$.
		
	\end{proof}

\end{document}